\RequirePackage{snapshot}       % Needed by bundledoc
\documentclass[11pt,a4paper,english]{article}
\usepackage{fullpage}

\usepackage[T1]{fontenc}
\usepackage{color}
\usepackage{babel}
\usepackage{refstyle}
\usepackage{calc}
\usepackage{amsthm}
\usepackage{amsmath}
\usepackage{amssymb}
\usepackage{mathtools}
\usepackage{esint}
\usepackage{xspace}
\usepackage{xargs}
\makeatletter

\AtBeginDocument{\providecommand\subref[1]{\ref{sub:#1}}}
\AtBeginDocument{\providecommand\secref[1]{\ref{sec:#1}}}
\AtBeginDocument{\providecommand\thmref[1]{\ref{thm:#1}}}
\AtBeginDocument{\providecommand\propref[1]{\ref{prop:#1}}}
\AtBeginDocument{\providecommand\claimref[1]{\ref{claim:#1}}}
\AtBeginDocument{}
\AtBeginDocument{\providecommand\lemref[1]{\ref{lem:#1}}}
\AtBeginDocument{\providecommand\apxref[1]{\ref{apx:#1}}}
\pdfpageheight\paperheight
\pdfpagewidth\paperwidth

  \theoremstyle{definition}
  \newtheorem{defn}{\protect\definitionname}
\theoremstyle{plain}
\newtheorem{thm}{\protect\theoremname}
  \theoremstyle{plain}
  \newtheorem{prop}{\protect\propositionname}
  \theoremstyle{plain}
  
  \theoremstyle{remark}
  \newtheorem{claim}{\protect\claimname}
  \theoremstyle{plain}
  \newtheorem{lem}{\protect\lemmaname}

\@ifundefined{date}{}{\date{}}
\usepackage{babel}
\usepackage{bbm}

  \providecommand{\claimname}{Claim}
  \providecommand{\definitionname}{Definition}
  
  \providecommand{\lemmaname}{Lemma}
  \providecommand{\notename}{Note}

\providecommand{\theoremname}{Theorem}

\newtheorem*{note*}{\protect\notename}

\newcommand{\model}{uniform push model\xspace}

\newcommand{\nfrac}[2]{#1/#2}
\newcommand{\totballs}{h}
\newcommand{\colballs}[1]{h_{#1}}

\newcommand{\longconfig}{opinion distribution\xspace}
\newcommand{\config}{opinion distribution\xspace}
\newcommand{\bc}{\mathbf{c}}

\newcommand{\st}{such that\xspace}
\newcommand{\wrt}{with respect to\xspace}

\newcommand{\ppapx}{process $\mathbf{P}$\xspace}
\newcommand{\porig}{process $\mathbf{O}$\xspace}
\newcommandx{\rec}[2][1=u,2=j]{R_{#2}(#1)}
\newcommand{\M}[1][j]{\ensuremath{M_{#1}}\xspace}
\newcommand{\N}[1][j]{\ensuremath{N_{#1}}\xspace}

\newcommand{\pbib}{process $\mathbf{B}$\xspace}

\newcommand{\realizsetN}{\mathcal{N}}

\newcommand{\bigo}{{\mathcal O}}

\newcommand{\maj}[1]{\mathrm{maj}(#1)}
\newcommand{\mode}[1]{\mathrm{mode}(#1)}

\newcommand{\Sam}[1]{\mathcal{S}(#1)}

\newcommand{\eventhighprob}{ \mathcal{E} }
\newcommand{\ceventhighprob}{ \bar\eventhighprob}

\newcommand{\constpoisapx}{b}
\newcommand{\constfour}{c}

\newcommand{\pc}{j}

\newref{claim}{name=Claim~}
\newref{prop}{name=Proposition~}
\newref{def}{name=Definition~}
\newref{thm}{name=Theorem~}
\newref{sec}{name=Section~}
\newref{sub}{name=Section~}
\newref{lem}{name=Lemma~}
\newref{apx}{name=Appendix~}
\newref{cor}{name=Corollary~}

\makeatother

  \providecommand{\claimname}{Claim}
  \providecommand{\definitionname}{Definition}
  \providecommand{\lemmaname}{Lemma}
  \providecommand{\propositionname}{Proposition}
\providecommand{\corollaryname}{Corollary}
\providecommand{\theoremname}{Theorem}

\begin{document}

\title{Noisy Rumor Spreading and Plurality Consensus}

\author{Pierre Fraigniaud\thanks{Additional support from the ANR project
    DISPLEXITY and from the INRIA project GANG.}\\
 {\small{} Laboratoire d'Informatique Algorithmique}\\
 {\small{}CNRS and University Paris Diderot}\\
 {\small{}France} \and Emanuele Natale\thanks{This work was done in part while
 the author was visiting the Simons Institute for the Theory of Computing.}\\
 {\small{}Dipartimento di Informatica}\\
 {\small{}Sapienza Universit\`a di Roma}\\
 {\small{}Italy} }
 
\maketitle

\begin{abstract}
Error-correcting codes are efficient methods for handling \emph{noisy}
communication channels in the context of technological networks. However,
such elaborate methods differ a lot  from the unsophisticated
way biological entities are supposed to communicate. Yet, it has been recently shown
by Feinerman, Haeupler, and Korman {[}PODC 2014{]} that complex coordination
tasks such as \emph{rumor spreading} and \emph{majority consensus}
can plausibly be achieved in biological systems subject to noisy communication
channels, where every message transferred through a channel remains
intact with small probability $\frac{1}{2}+\epsilon$, without using
coding techniques. This result is a considerable step towards a better
understanding of the way biological entities may cooperate. It has nevertheless 
been established only in the case of 2-valued \emph{opinions}: rumor
spreading aims at broadcasting a single-bit opinion to all nodes,
and majority consensus aims at leading all nodes to adopt the single-bit
opinion that was initially present in the system with (relative) majority.
In this paper, we extend this previous work to $k$-valued opinions,
for any constant $k\geq2$.

Our extension requires to address a series of important issues, some
conceptual, others technical. We had to  entirely revisit the notion of noise,
for handling channels carrying $k$-\emph{valued} messages. In fact, we
precisely characterize the type of noise patterns for which plurality consensus
is solvable. Also, a key result employed in the bivalued case by Feinerman et
al. is an estimate of the probability of observing the most frequent opinion
from observing  the mode of a small sample. We generalize this result to the
multivalued case by providing a new analytical proof for the bivalued case that
is amenable to be extended, by induction, and that is of independent interest.
\end{abstract}
\bigskip

%%%%%%%%%%%%%%%%%%%%%%%%%%%%%%%%%%%%%%%%%%%%%%%

\section{Introduction}

%%%%%%%%%%%%%%%%%%%%%%%%%%%%%%%%%%%%%%%%%%%%%%%

\subsection{Context and Objective}
\label{sub:context}

%---------------------------------------------

To guarantee reliable communication over a network in the presence
of noise is the main goal of Network Information Theory~\cite{el2011network}.
Thanks to the achievements of this theory, the impact of noise can
often be drastically reduced to almost zero by employing \emph{error-correcting
codes}, which are practical methods whenever dealing with artificial
entities. However, as observed in~\cite{FHK14},
the situation is radically different for scenarios
in which the computational entities are biological. Indeed, from a
biological perspective, a computational process can be considered
``simple'' only if it consists of very basic primitive operations, and
is extremely lightweight. As a consequence, it is unlikely that biological
entities are employing techniques like error-correcting codes to reduce
the impact of noise in communications between them. Yet, biological
signals are subject to noise, when generated, transmitted, and received.
This rises the intriguing question of how entities in biological ensembles
can cooperate in presence of noisy communications, but in absence
of mechanisms such as error-correcting codes.

An important step toward understanding communications in biological
ensembles has been achieved recently in~\cite{FHK14},
which showed how it is possible to cope with noisy communications
in absence of coding mechanisms for solving complex tasks such as
\emph{rumor-spreading} and \emph{majority consensus}. Such a result
provides highly valuable hints on how complex tasks can be achieved
in frameworks such as the immune system, bacteria populations, or
super-organisms of social insects, despite the presence of noisy communications.

In the case of rumor-spreading, \cite{FHK14}
assumes that a source-node initially handles a bit, set to some binary
value, called the \emph{correct opinion}. This opinion has to be transmitted
to all nodes, in a noisy environment, modeled as a complete network
with unreliable links. More precisely, messages are transmitted in
the network according to the classical \emph{\model}
~\cite{Demers:1987:EAR:41840.41841,karp2000randomized,Pittel:1987:SR:37387.37400}
where, at each round, every node can send one binary opinion to a
neighbor chosen uniformly and independently at random and, before
reaching the receiver, that opinion is flipped with probability at
most $\frac{1}{2}-\epsilon$ with $\epsilon>0$. It is proved that,
even in this very noisy setting, the rumor-spreading problem can be
solved quite efficiently. Specifically, \cite{FHK14}
provides an algorithm that solves the noisy rumor-spreading problem
in $O(\frac{1}{\epsilon^{2}}\log n)$ communication rounds,
with high probability\footnote{A series of events $\mathcal{E}_{n}$, $n\geq 1$, hold w.h.p. if 
$\Pr(\mathcal{E}_{n})\ge1-O(1/n^{c})$ for some $c>0$.} (w.h.p.) in $n$-node networks, using $O(\log\log n+\log(1/\epsilon))$
bits of memory per node. Again, this algorithm exchanges solely opinions
between nodes. 

In the case of majority consensus, \cite{FHK14}
assumes that some nodes are supporting opinion~0, some nodes are
supporting opinion~1, and some other nodes are supporting no opinion.
The objective is that all nodes eventually support the initially most
frequent opinion (0 or 1). More precisely, let $A$ be the set of
nodes with opinion, and let $b\in\{0,1\}$ be the majority opinion
in $A$. The \emph{majority bias} of $A$ is defined as $\frac{1}{2}(|A_{b}|-|A_{\bar{b}}|)/|A|$
where $A_{i}$ is the set of nodes with opinion~$i\in\{0,1\}$. In
the very same noisy communication model as above, \cite{FHK14}
provides an algorithm that solves the noisy majority consensus problem
for $|A|=\Omega(\frac{1}{\epsilon^{2}}\log n)$ with majority-bias
$\Omega(\sqrt{\log n/|A|})$. The algorithm runs in $O(\frac{1}{\epsilon^{2}}\log n)$
rounds, w.h.p., in $n$-node networks, using $O(\log\log n+\log(1/\epsilon))$
bits of memory per node. As for the case of rumor spreading, the algorithm
exchanges solely opinions between nodes. In fact, the latter algorithm
for majority consensus is used as a subroutine for solving the rumor-spreading problem. 
Note that the majority consensus algorithm of \cite{FHK14}
requires that the nodes are initially aware of the size of $A$.

According to \cite{FHK14}, both algorithms
are optimal, since both rumor-spreading and majority consensus require
$\Omega(\frac{1}{\epsilon^{2}}\log n)$ rounds w.h.p.\ in
$n$-node networks.

Our objective is to extend the work of ~\cite{FHK14}
to the natural case of an arbitrary number of opinions, to go beyond 
a proof of concept. The problem that results from this extension is an instance
of the \emph{plurality consensus} problem in the presence of noise, i.e.,
the problem of making the system converging to the initially most frequent
opinion (i.e., the \emph{plurality} opinion). Indeed, the plurality consensus 
problem  naturally arises in several  biological settings, 
typically for choosing between different directions for a
flock of birds \cite{Dolev}, different speeds for a school of fish
\cite{FishConsensus}, or different nesting sites for ants \cite{HouseHunt}.
The computation of the most frequent value has also been observed
in biological cells \cite{cardelli2012cell}.

\subsection{Our Contribution}

%---------------------------------------------

\subsubsection{Our Results}

We generalize the results in \cite{FHK14}
to the setting in which an arbitrary large number $k$ of opinions is
present in the system. In the context of rumor spreading, the correct
opinion is a value $i\in\{1,\dots,k\}$, for any constant $k\geq2$.
Initially, one node supports this opinion $i$, and the other nodes
have no opinions. The nodes must exchange opinions so that, eventually,
all nodes support the correct opinion~$i$. In the context of (relative)
majority consensus, also known as \emph{plurality consensus}, each
node $u$ initially supports one opinion $i_{u}\in\{1,\dots,k\}$,
or has no opinion. The objective is that all nodes eventually adopt
the \emph{plurality opinion} (i.e., the opinion initially held by
more nodes than any other, but not necessarily by an overall majority
of nodes). As in \cite{FHK14}, we restrict
ourselves to ``natural'' algorithms \cite{chazelle}, that is, algorithms
in which nodes only exchange opinions in a straightforward manner
(i.e., they do not use the opinions to encode, e.g., part of their
internal state). For both problems, the difficulty comes from the
fact that every opinion can be modified during its traversal of any
link, and switched at random to any other opinion. In short, we prove
that there are algorithms solving the noisy rumor spreading
problem and the noisy plurality consensus problem for multiple opinions,
with the same performances and probabilistic guarantees as the algorithms
for binary opinions in~\cite{FHK14}.

%---------------------------------------------

\subsubsection{The technical challenges}

%---------------------------------------------

Generalizing noisy rumor spreading and noisy majority consensus to
more than just two opinions requires to address a series of issues,
some conceptual, others technical. 

Conceptually, one needs first to redefine
the notion of noise. In the case of binary opinions, the noise can
just flip an opinion to its complement. In the case of multiple opinions,
an opinion $i$ subject to a modification is switched to another opinion
$i'$, but there are many ways of picking $i'$. For instance, $i'$
can be picked uniformly at random (u.a.r.) among all opinions. Or,
$i'$ could be picked as one of the ``close opinions'', say, either
$i+1$ or $i-1$ modulo~$k$. Or, $i'$ could be ``reset'' to,
say, $i=1$. In fact, there are very many alternatives, and
not all enable rumor spreading and plurality consensus to be solved.
One of our contributions is to characterize \emph{noise matrices}
$P=(p_{i,j})$, where $p_{i,j}$ is the probability that opinion $i$
is switched to opinion $j$, for which these two problems are efficiently
solvable. Similar issues arise for, e.g., redefining the majority
bias into a \emph{plurality bias}.

The technical difficulties are manifold. 
A key ingredient of the analysis in \cite{FHK14}
is a fine estimate of how nodes can mitigate the impact of noise by
observing the opinions of \emph{many }other nodes, and then considering
the mode of such sample.
Their proof relies on the fact that for the binary opinion case, given
a sample of size $\gamma$, the number of 1s and 0s in the sample
sum up to $\gamma$. Even for the ternary opinion case, the additional
\emph{degree of freedom }in the sample radically changes the nature
of the problem, and the impact of noise is statistically far more difficult to handle.

Also, to address the multivalued case, we had to cope with the fact that, in
the \model, the messages received by nodes at every round are
correlated. To see why, consider an instance of the system in which a certain
opinion $b$ is held by one node only, and there is no noise at all. In one
round, only one other node can receive $b$. It follows that if a certain node
$u$ has received $b$, no other nodes have received it. Thus, the messages each
node receives are not independent. In \cite{FHK14}, probability concentration
results are claimed for random variables (r.v.) that depend on such messages,
using Chernoff bounds. However, Chernoff bounds have been proved to hold only
for random variables whose stochastic dependence belongs to a very limited
class (see for example \cite{dubhashi1998concentration}). In \cite{FHK15}, it
is stated that the binary random variables on which the Chernoff bound is
applied satisfy the property of being negatively 1-correlated (see Section 1.7
in \cite{FHK15} for a formal definition).  
In our analysis, we show instead how to obtain concentration of probability in
this dependent setting by leveraging Poisson approximation techniques. Our
approach has the following advantage: instead of showing that the Chernoff
bound can be directly applied to the specific involved random variables, we
show that the execution of the given protocol, on the \model, can be
tightly approximated with the execution of the same protocol on a another
suitable communication model, that is not affected by the stochastic
correlation that affects the \model.
%---------------------------------------------
\subsection{Other Related work\label{sub:related}}

By extending the work of \cite{FHK14},
we contribute to the theoretical understanding of how communications
and interactions function in biological distributed systems, from an
algorithmic perspective 
\cite{AAER07,AR07,beeping1,beeping2,AAE07,korman2014confidence}.
We refer the reader to \cite{FHK14}
for a discussion on the computational aspects of biological distributed
systems, an overview of the rumor spreading problem in distributed
computing, as well as its biological significance in the presence of noise.
In this section, we mainly discuss the previous technical contributions from the literature
 related to the  novelty of our work, that is the extension to the case of several
different opinions. 

 In the context of population
protocol, the problem of achieving majority consensus in the binary
case has been solved by employing a simple protocol called \emph{undecided
state dynamic} \cite{AAE07}. 
In the \model, the binary
majority consensus problem can be solved very efficiently as a consequence
of a more general result about computing the median of the initial
opinions \cite{DGMSS11}. Still in the \model, the undecided
state dynamic has been analyzed in the case of an arbitrarily
large number of opinions, which may even be a function of the number
of agents in the system \cite{Becchetti:2015:PCG:2722129.2722156}. A similar result
has been obtained for another elementary protocol, so-called \emph{3-majority
dynamics}, in which, at each round, each node samples the opinion of
three random nodes, and adopts the most frequent opinion among these
three \cite{BCNPST13}. The 3-majority dynamics has also been shown to be fault-tolerant
against an adversary that can change up to $O(\sqrt{n})$
agents at each round \cite{DGMSS11, BCNPST13}. Other works have analyzed the
undecided state dynamics in asynchronous models with a constant number
of opinions \cite{DV12,jung2012distributed,PVV09}, and the $h$-majority
dynamics (or slight variations of it) on different graph classes
in the \model \cite{CER14,AD12}.
Note that the above protocols converge to the correct value with high 
probability. In \cite{AGV15}, the tradeoff between deterministic success and
convergence speed for protocols solving the binary majority consensus problem
in population protocols is investigated.

A general result by Kempe et al.~\cite{KDG03} shows how to compute a large class
of functions in the \model. However, their
protocol requires the nodes to send slightly more complex messages
than their sole current opinion, and its effectiveness heavily relies
on a potential function argument that does not hold in the presence
of noise. 

To the best of our knowledge, we are the first considering the plurality
consensus problem in the presence of noise. 

%%%%%%%%%%%%%%%%%%%%%%%%%%%%%%%%%%%%%%%%%%%%

\section{Model and formal statement of our results}
\label{sec:The-Model} 

In this section we formally define the communication model, the main definitions, the 
investigated problems and our contribution to them. 

We do not provide a definition of what is a biologically feasible protocol,
since the computational investigation with this respect is still too premature
for such an attempt. As discussed in Section \ref{sub:context}, intuitively we
look for protocols that are simple enough to be plausible communication
strategies for primitive biological system. As the reader can see in section
\ref{sub:formal_results} and \ref{sub:protocol}, we consider a natural
generalization of the protocol given in \cite{FHK15}, which is plainly an
elementary combination of sampling and majority operations.

\subsection{Communication model and definition of the problems}

The communication model we consider is essentially the \emph{\model}~\cite{Demers:1987:EAR:41840.41841,karp2000randomized,Pittel:1987:SR:37387.37400},
where in each (synchronous) round each agent can send (\emph{push)} a message
to another agent chosen uniformly at random. This occurs without having the
sender or the receiver learning about each other's identity. Note that it may
happen that several agents push a message to the same node $u$ at the same
round. In the latter case we assume that the nodes receive them in a random
order; for a detailed discussion regarding this assumption, we refer the reader
to Appendix \ref{apx:simultaneous_messages}.

We study the problems of rumor-spreading and plurality consensus. In both
cases, we assume that nodes can support opinions represented by an integer in
$[k]=\{1,\dots,k\}$. Additionally, there may be \emph{undecided} nodes that do
not support any opinion, which represents nodes that are not actively aware
that the system has started to solve the problem; thus, undecided nodes are not
allowed to send any message before receiving any of them.
\begin{itemize}
\item In rumor spreading, initially, one node, called the source, has an
opinion $m\in\{1,\dots,k\}$, called the \emph{correct opinion}. All
the other nodes have no opinion. The objective is to design a protocol
insuring that, after a certain number of communication rounds, every
node has the correct opinion~$m$. 
\item In plurality consensus, initially, for every $i\in\{1,\dots,k\}$, a set
    $A_{i}$ of nodes have opinion~$i$. The sets $A_{i}$, $i=1,\dots,k$, are
    pairwise disjoint, and their union does not need to cover all nodes, i.e.,
    there may be some \emph{undecided} nodes with no opinion initially. The
    objective is to design a protocol insuring that, after a certain number of
    communication rounds, every node has the plurality opinion, that is, the
    opinion $m$ with relative majority in the initial setting (i.e.,
    $|A_{m}|>|A_{j}|$ for any $j\neq m$). 
\end{itemize}

Observe that the rumor-spreading problem is a special case of the plurality
consensus problem with $|A_{m}|=1$ and $|A_{j}|=0$ for any $j\neq m$.

Following the guidelines of \cite{FHK14},
we work under two constraints: 
\begin{enumerate}
\item \label{enu:msgs_in_1..k} We restrict ourselves to protocols in which
each node can only transmit opinions, i.e., every message is 
an integer in $\{1,\dots,k\}$. 
\item \label{enu:noise_in_msgs} Transmissions are subject to noise, that
is, for every round, and for every node $u$, if an opinion $i\in\{1,\dots,k\}$
is transmitted to node $u$ during that round, then node $u$ will
receive message $j\in\{1,\dots,k\}$ with probability $p_{i,j}\geq0$,
where $\sum_{j=1}^{k}p_{i,j}=1$. 
\end{enumerate}
The \emph{noisy push model } is the \model together with
the previous two constraints.
The probabilities $\{p_{i,j}\}_{i,j\in[k]}$ can be seen as a transition
matrix, called the \emph{noise matrix}, and denoted by $P=(p_{i,j})_{i,j\in[k]}$.
The noise matrix in \cite{FHK14} is
simply 
\begin{equation}
P=\left(\begin{array}{cc}
\frac{1}{2}+\epsilon & \frac{1}{2}-\epsilon\\
\frac{1}{2}-\epsilon & \frac{1}{2}+\epsilon
\end{array}\right).\label{eq:binary_n_m}
\end{equation}

%---------------------------------------------

\subsection{Plurality bias, and majority preservation}

%---------------------------------------------

When time proceeds, our protocols will result in the proportion of nodes with a
given opinion to evolve. Note that there might be nodes who do not support any
opinion at time $t$. As mentioned in the previous section, we call such nodes
\emph{undecided}. We denote by $a^{(t)}$ the fraction of nodes supporting any
opinion at time $t$ and we call the nodes contributing to $a^{(t)}$
\emph{opinionated}. Consequently, the fraction of undecided nodes at time $t$
is $1-a^{(t)}$. Let $c_{i}^{(t)}$ be the fraction of opinionated nodes in the
system that support opinion $i\in[k]$ at the beginning of round $t$, so that
$\sum_{i\in [k]} c_{i}^{(t)} = a^{(t)}$. Let $\hat{c}_{i}^{(t)}$ be the
fraction of opinionated nodes which receive at least one message at time $t-1$
and support opinion $i\in[k]$ at the beginning of round $t$. 
We write $ \mathbf{c}^{(t)}=(c_{1}^{(t)},...,c_{k}^{(t)})$ to denote the
\emph{\longconfig} of the opinions at time $t$. Similarly, let
$\hat{\mathbf{c}}^{(t)}=(\hat{c}_{1}^{(t)},...,\hat{c}_{k}^{(t)})$. In
particular, if every node would simply switch to the last opinion it received,
then 
\begin{equation*}
\mathbb{E}[\hat{c}_{i}^{(t+1)}\mid\mathbf{c}^{(t)}]  = 
\sum_{j\in[k]}\Pr[\mbox{received \ensuremath{i}}\mid\mbox{original message is \ensuremath{j}}]\cdot\Pr[\mbox{original message is \ensuremath{j}}] =\sum_{j\in[k]}c_{j}^{(t)}\cdot p_{j,i}.
\end{equation*}
That is,
\begin{equation}
	\mathbb{E}[\hat{\mathbf{c}}^{(t+1)}
	\mid\mathbf{c}^{(t)}]=\mathbf{c}^{(t)}\cdot P,
	\label{obs:average-intuition-P}
\end{equation}
where $P$ is the noise matrix. In particular, in the absence of noise, we have
$P=I$ (the identity matrix), and if every node would simply copy the opinion
that it just received, we had $
\mathbb{E}[\hat{\mathbf{c}}^{(t+1)}\mid\mathbf{c}^{(t)}]=\mathbf{c}^{(t)}$. So,
given the \config at round $t$, from the definition of the model it follows
that the messages each node receives at round $t+1$ can equivalently be seen as
being sent from a system without noise, but whose \config at round $t$ is
$\mathbf{c}^{(t)}\cdot P$. 

Recall that $m$ denotes the initially correct opinion, that is, the source's
opinion in the rumor-spreading problem, and the initial plurality
opinion in the plurality consensus problem. The following definition
naturally extends the concept of \emph{majority bias} in \cite{FHK14}
to \emph{plurality bias}. 
\begin{defn}
\textit{Let $\delta>0$. An \config $\mathbf{c}$ is said to be} {\rm $\delta$-biased
toward opinion $m$}\textit{ if 
$
c_{m}-c_{i}\geq\delta
$
for all $i\neq m$. }
\end{defn}
In \cite{FHK14}, each binary opinion
that is transmitted between two nodes is flipped with probability
at most $\frac{1}{2}-\epsilon$,
with\footnote{For a discussion on what happens for other values of $\epsilon$, see
\apxref{eps-smaller}.} 
$\epsilon=n^{-\frac{1}{4}+\eta}$ for an arbitrarily small
$\eta>0$. Thus, the noise is parametrized by $\epsilon$. The smaller
$\epsilon$, the more noisy are the communications. We generalize the role of
this parameter with the following definition. 

\begin{defn} 
    \label{def:noise}
	\textit{Let $\epsilon=\epsilon(n)$ and $\delta=\delta(n)$ be positive. 
    A noise matrix $P$ is said to be}
	{\rm $(\epsilon, \delta)$-majority-preserving ($(\epsilon, \delta)$-m.p.) \wrt
	opinion $m$} \textit{if, for  every \config
	$\mathbf{c}$ that is $\delta$-biased toward opinion~$m$, we have 
%	\begin{equation} 
$
		\left(\mathbf{c}\cdot
		P\right)_{m}-\left(\mathbf{c}\cdot
		P\right){}_{i}>\epsilon\,\delta
		%\label{eq:eps-maj-preserving}
$
%	\end{equation}
	for all $i\neq m$. }
\end{defn}

In the rumor-spreading problem, as well as in the plurality consensus problem,
when we say that a noise matrix is $(\epsilon, \delta)$-m.p., we implicitly
mean that it is  $(\epsilon, \delta)$-m.p. \wrt the initially correct opinion.
Because of the space constraints, we defer the discussion on the class of
$(\epsilon, \delta)$-m.p. noise matrices in \secref{mp_matrices} (including its
tightness w.r.t. theorems \ref{thm:general_rumor_spreading} and
\ref{thm:general_majority_consensus}).

%---------------------------------------------

\subsection{Our formal results}
\label{sub:formal_results}

%---------------------------------------------

We show that a natural generalization of the protocol in \cite{FHK14} solves
the rumor spreading problem and the plurality consensus problem for an
arbitrary number of opinions $k$. More precisely, using the protocol which we
describe in \subref{protocol}, we can establish the following two results,
whose proof can be found in \secref{The-Analysis}.

\begin{thm}
\label{thm:general_rumor_spreading} Assume that the noise matrix
\textbf{$P$} is $(\epsilon, \delta)$-m.p. with $\epsilon=\Omega(n^{-\frac{1}{4}+\eta})$
for an arbitrarily small constant $\eta>0$ and $\delta=\Omega(\sqrt{\nfrac{\log n}n})$. The noisy rumor-spreading problem
with $k$ opinions can be solved in $O(\frac{\log n}{\epsilon^{2}})$
communication rounds, w.h.p., by a protocol using $O(\log\log n+\log\frac{1}{\epsilon})$
bits of memory at each node. 
\end{thm}

\begin{thm}
\label{thm:general_majority_consensus} 
Let $S$ with $|S|=\Omega(\frac{1}{\epsilon^{2}}\log n)$ be an
initial set of nodes with opinions in $[k]$, the rest of the nodes having no
opinions. Assume that the noise matrix \textbf{$P$} is
$(\epsilon, \delta)$-m.p. for some $\epsilon>0$, and that $S$ is
$\Omega(\sqrt{\log n/|S|})$-majority-biased.
The noisy plurality consensus problem with $k$ opinions can be solved
in $O(\frac{\log n}{\epsilon^{2}})$ communication rounds,
w.h.p., by a protocol using $O(\log\log n+\log\frac{1}{\epsilon})$
bits of memory at each node. 
\end{thm}

For $k=2$, we get the theorems in \cite{FHK14} from the above two theorems.
Indeed, the simple 2-dimensional noise matrix of Eq.~(\ref{eq:binary_n_m})
is $\epsilon$-majority-biased. Note that, as in~\cite{FHK14},
the plurality consensus algorithm requires the nodes to known the
size $|S|$ of the set $S$ of opinionated nodes. 

%%%%%%%%%%%%%%%%%%%%%%%%%%%%%%%%%%%%%%%%%%%%

\section{The Analysis\label{sec:The-Analysis}}

%%%%%%%%%%%%%%%%%%%%%%%%%%%%%%%%%%%%%%%%%%%%

In this section we prove \thmref{general_rumor_spreading} and \thmref{general_majority_consensus}
by generalizing the analysis of Stage 1 given in \cite{FHK14}
and by providing a new analysis of Stage 2. 
Note that the proof techniques required for the generalization to arbitrary $k$ significantly
depart from those in \cite{FHK14} for
the case $k=2$. In particular, our approach provides a general framework to
rigorously deal with many kind of stochastic dependences among messages in the
\model.

\subsection{Definition of the Protocol}
\label{sub:protocol}

We describe a rumor spreading protocol performing in two
\emph{stages}. Each stage is decomposed into a number of \emph{phases},
each one decomposed into a  number of \emph{rounds}. During each phase of
the two stages, the nodes apply the simple rules given below.

\subsubsection{The rule during each phase of Stage 1.} Nodes that already support some
opinion at the beginning of the phase push their opinion at each round of the
phase. Nodes that do not support any opinion at the beginning of the phase but
receive at least one opinion during the phase start supporting an opinion at
the end of the phase, chosen u.a.r. (counting multiplicities) from the received
opinions\footnote{\label{fn:uar}Note that, in the protocol considered in
\cite{FHK14}, the choice of each node's new opinion in both stages is based on
the first messages received. In \cite{FHK15}, in order to relax the
synchronicity assumption that nodes share a common clock, they adopt the same
sample-based variant of the rule that we adopt here.}. In other words, each
node tries to acquire an opinion during each phase of Stage~1, and, as it
eventually receives some opinions, it starts supporting one of them (chosen
u.a.r.) from the beginning of the next phase. In particular, opinionated nodes
never change their opinion during the entire stage.

More formally, let $\phi,\beta$, and $s$ be three constants satisfying $\phi>\beta> s$.
The rounds of Stage 1 are grouped in $T+2$ phases with 
$ T=\lfloor \log(n/(2s/\epsilon^{2} \log n))/\log(\beta/\epsilon^{2}+1)\rfloor $.
Phase 0 takes $s/\epsilon^{2} \log n$ rounds, phase $T+1$ takes $\phi/\epsilon^{2} \log n$
rounds, and each phase $\pc $ with $1\leq\pc \leq T$ takes
$\beta/\epsilon^{2}$ rounds. We denote with $\tau_{\pc}$ the end of the last
round of phase $\pc$.

Let $t_{u}$ be the first time in which $u$ receives any opinion since the
beginning of the protocol (with $t_u=0$ for the source). Let $\pc_{u}$ be
the phase of $t_{u}$, and let $\mbox{val}(u)$ be an opinion chosen u.a.r. by
$u$ among those that it receives during phase $\pc_u$\footnote{Note that,
in order to sample u.a.r. one of them, $u$ does not need to collect all the
opinions it receives. A natural sampling strategy such as reservoir sampling
can be used.}. During the first stage of the protocol each node applies the
following rule.

\smallskip
\fbox{\begin{minipage}[t]{\columnwidth}%
\textbf{Rule of Stage 1. }Each opinionated node $u$ pushes opinion $\mbox{val}(u)$
during each round of every phase $\pc=\pc_{u}+1,...,T+1$. %
\end{minipage}}

\subsubsection{The rule during each phase of Stage 2.}
During each phase of Stage~2, every node pushes its opinion at each
round of the phase. At the end of the phase, each node that received
``enough'' opinions takes a random sample\textsuperscript{\ref{fn:uar}}
of them, and starts supporting the most frequent opinion in that sample
(breaking ties u.a.r.). 

More formally, the rounds of stage 2 are divided in $T^{\prime}+1$ phases with $T^{\prime}=\lceil \log(\sqrt{n/\log n}) \rceil$. 
Each phase $j$, $0\leq j\leq T^{\prime}-1$, has length $2\ell$ with
$\ell= \lceil \constfour/\epsilon^{2} \rceil $ for
some large-enough constant $\constfour>0$, and phase $T^{\prime}$
has length $2\ell'$ with $\ell'=O(\epsilon^{-2}\log n)$.
For any finite multiset $A$ of elements in $\{1,\dots,k\}$, and any $i\in\{1,\dots,k\}$,  let $\mbox{occ}(i,A)$ be the number of occurrences of $i$ in $A$, and let 
$
	\mode A= \{ i \in \{1,\dots,k\} \mid \mbox{occ}(i,A) \geq  \mbox{occ}(j,A) \; \mbox{for every $j \in \{1,\dots,k\}$}\}. 
$
We then define $\maj A$ as the most frequent value in $A$
(breaking ties u.a.r.), i.e., $\maj A$ is the r.v. on $\{1,\dots,k\}$ \st
$ \Pr( \maj A = i ) ={\mathbbm1_{\left\{ i\in\mode A\right\} }}/{\left|\mode A\right|}. $
Let $\rec$ be the multiset of messages received by node $u$ during phase $j$.
During the second stage of the protocol each node applies the following rule.

\smallskip
\fbox{\begin{minipage}[t]{\columnwidth}%
\textbf{Rule of Stage 2. }During each phase $j$ of length $2L$ of Stage 2
($L=\ell$ or $\ell'$), each node $u$ pushes its current opinion at each round
of the phase, and starts drawing a random uniform sample $\Sam{u}$ of size $L$
from $\rec$. Provided $|\rec|\geq L$, at the end of the phase $u$ changes its
opinion to $\maj{\Sam{u}}$. 
\end{minipage}}

\smallskip
Let us remark that the reason we require the use of sampling in the previous rule is that
at a given round a node may receive much more messages than $2L$. Thus, if the nodes were to 
collect all the messages they receive, some of them would need much more memory than
our protocol does. 
Finally, observe that overall both stages 1 and 2 take $O(\frac{1}{\epsilon^{2}}\log n)$
rounds. 

\subsection{Pushing Colored Balls into Bins\label{sub:push-balls}}

Before delving into the analysis of the protocol, we provide a framework to
rigorously deal with the stochastic dependence that arises between messages in
the \model. 
Let \porig be the process that results from the execution of the protocol of
\subref{protocol} in the \model. In order to apply concentration of
probability results that requires the involved random variables to be
independent, we view the messages as balls, and the nodes as bins, and employ
Poisson approximation techniques.  More specifically, during each phase $j$ of
the protocol, let $\M$ be the set of messages that are sent to random nodes,
and $\N$ be the set of  messages sent \emph{after} the noise has acted on them.
(In other words,  $\N= \bigcup_{u}\rec$).  We prove that, at the end of
phase~$j$, we can equivalently assume that all the messages \M have been sent
to the nodes according to the following process.

\begin{defn}\label{def:bib}
    The balls-into-bins \pbib associated to phase $j$ is the two-step process
    in which the nodes represent bins and all messages sent in the phase
    represent colored balls, with each color corresponding to some opinion.
    Initially, balls are colored according to $\M$. At the first step, each
    ball of color $i\in \{1,\dots,k\}$ is re-colored with color $j \in
    \{1,\dots,k\}$ with probability $ p_{i,j}$, independently of the other
    balls. At the second step all balls are thrown to the bins u.a.r. as in a
    balls-into-bins experiment.
\end{defn}

\begin{claim}
\label{claim:prot_eq_bib}
	Given the \config and the number of active nodes 
    at the beginning of phase $j$, the probability 
	distribution of the \config and the number of active nodes
    at the end of phase $j$ in \porig is the same as if the messages were sent
	according to \pbib.
\end{claim}

It is not hard to see that \claimref{prot_eq_bib} holds in the case of a single
round. For more than one round, it is crucial to observe that the way each node
$u$ acts in the protocol depends only on the received messages $\rec$,
regardless of the order in which these messages are received. As an example,
consider the \config in which one node has opinion~1, one other node has
opinion~2, and all other nodes have opinion~3. Suppose that each node pushes
its opinion for two consecutive rounds. Since, at each round, exactly one
opinion~1 and exactly one opinion~2 are pushed, no node can receive two 1s
during the first round and then two 2s during the second round, i.e. no node
can possibly receive the sequence of messages ``1,1,2,2'' in this exact order.
Instead, in \pbib such a sequence is possible. 

\begin{proof}[Proof of \claimref{prot_eq_bib}]
    In both \pbib and \porig, at each round, the noise acts independently on
    each ball/message of a given color/opinion, according to the same
    probability distribution for that color/opinion. 
    Then, in both processes, each ball/message is sent to some bin/node chosen
    u.a.r. and independently of the other balls/messages. 
    Indeed, we can couple \pbib and \porig by requiring that: 
    \begin{enumerate}
        \item each ball/message is changed by the noise to the same color/value, and 
        \item each ball/message ends up in the same bin/node.
    \end{enumerate}
    Thus, the joint probability distribution of the sets 
    ${\left\{ \rec \right\}}_{u\in [n]}$
    in \porig is the same as the one given by \pbib.

    Observe also that, from the definition of the protocol (see the rule of
    Stage 1 and Stage 2 in \subref{protocol}), it follows that each node's
    action depends only on the set $\rec$ of received messages at the end of
    each phase~$j$, and does not depend on any further information such as the
    actual order in which the messages are received during the phase. 

	Summing up the two previous observations, we get that
	if, at the end of each phase $j$, we generate the $\rec$s according
	to \pbib, and we let the protocol execute according to them, then we 
	indeed get the same stochastic process as \porig.
\end{proof}

Now, one key ingredient in our proof is to approximate \pbib using the following \ppapx.

\begin{defn}\label{def:pois_prox}
    Given \N[j], \ppapx associated to phase $j$ is the one-shot process in
    which each node receives a number of opinions $i$ that is a random
    variable with distribution Poisson$(\colballs{i}/n)$, where $\colballs{i}$
    is the number of messages in \N[j] carrying opinion~$i$, and each
    Poisson random variable is independent of the others.
\end{defn}

Now we provide some results from the theory of Poisson approximation
for balls-in-bins experiments that are used in \subref{push-balls}.
For a nice introduction to the topic, we refer to \cite{Mitzenmacher:2005:PCR:1076315}.
\begin{lem}
\label{lem:pois-cond}Let $\left\{ X_{j}\right\} _{j\in\left[\tilde{n}\right]}$
be independent r.v. such that $X_{j}\sim$Poisson$(\lambda_{j})$.
The vector $(X_{1},...,X_{\tilde{n}})$ conditional on
$\sum_{j}X=\tilde{m}$ follows a multinomial distribution with $\tilde{m}$
trials and probabilities $(\frac{\lambda_{1}}{\sum_{j}\lambda_{j}},...,\frac{\lambda_{\tilde{n}}}{\sum_{j}\lambda_{j}})$.
\end{lem}

\begin{lem}
\label{lem:pois-apx}
Consider a balls-in-bins experiment in which
$\totballs$ colored balls are thrown in $n$ bins, where $\colballs{i}$ balls
have color $i$ with $i\in\left\{ 1,...,k\right\} $ and
$\sum_{i}\colballs{i}=\totballs$. Let $\left\{ X_{u,i}\right\}
_{\substack{u\in\left\{ 1,...,n\right\}, i\in\left\{ 1,...,k\right\} } }$
be the number of $i$-colored balls that end up in bin $u$, let
$f(x_{1,1},...,x_{n,1},x_{n,2},...,x_{n,n},z_{1},...,z_{n})$ be a
non-negative function with positive integer arguments $x_{1,1},...,x_{n,1},\!$
$x_{n,2},...,x_{n,n},\!$ $z_{1},...,z_{n}$,
let $\left\{ Y_{u,i}\right\} _{\substack{u\in\left\{ 1,...,n\right\},
i\in\left\{ 1,...,k\right\} } }$ be independent r.v. such that
$Y_{u,i}\sim$Poisson$(\colballs{i}/n)$
and let $Z_{1},...,Z_{n}$ be integer valued r.v. independent from
the $X_{u,i}$s and $Y_{u,i}$s. Then 
\begin{multline*}
\mathbb{E}\left[f\left(X_{1,1}, ..., X_{n, 1}, X_{n, 2}, ..., X_{n, n}, Z_{1},
..., Z_{n}\right)\right]\\ 
\leq e^{k}\sqrt{\prod_{i}\colballs{i}}\
\mathbb{E}\left[f\left(Y_{1,1}, ..., Y_{n,1},Y_{n,2}, ..., Y_{n,n},Z_{1}, ...,
Z_{n}\right)\right].
\end{multline*}
\end{lem}
\begin{proof}
To simplify notation, let $\bar{Z}=(Z_{1},...,Z_{n})$,
$\bar{X}=(X_{1,1},...,X_{n,1},X_{n,2},...,X_{n,n})$,
$\bar{Y}=(Y_{1,1},...,Y_{n,1},Y_{n,2},...,Y_{n,n})$,
$\bar{Y}_{\sum}=(\sum_{u=1}^{n}Y_{u,1},...,\sum_{u=1}^{n}Y_{u,k})$,
$\lambda_{i}=\colballs{i}/n$, $\bar{\lambda}=(\lambda_{1},...,\lambda_{k})$
and finally $\bar{x}=(x_{1},...,x_{k})$ for any $x_{1},...,x_{k}$. 
Observe that, while $X_{u,i}$ and $X_{v,i}$ are clearly dependent, $X_{u,i}$
and $X_{v,j}$ with $i\neq j$ are stochastically independent (even
if $u=v$). Indeed, the distribution of the r.v. 
$\left\{ X_{u,i}\right\} _{u\in\left\{ 1,...,n\right\} }$
for each fixed $i$ is multinomial with $\lambda_{i}$ trials and
uniform distribution on the bins. Thus, from \lemref{pois-cond} we have that
$\left\{ X_{u,i}\right\} _{u\in\left\{ 1,...,n\right\} }$are distributed
as $\left\{ Y_{u,i}\right\} _{u\in\left\{ 1,...,n\right\} }$ conditional
on $\sum_{u=1}^{n}Y_{u,i}=\lambda_{i}$, that is 
\begin{equation*}
\mathbb{E}\left[f\left(\bar{Y},\bar{Z}\right)\middle|\ 
\sum_{u=1}^{n}Y_{u,1}=\lambda_{1},...,\sum_{u=1}^{n}Y_{u,k}=\lambda_{k}\right] 
=\mathbb{E}\left[f\left(\bar{X},\bar{Z}\right)\right].
\end{equation*}
Therefore, we have 
\begin{align*}
\mathbb{E}\left[f\left(\bar{Y},\bar{Z}\right)\right]
&= \sum_{\bar{x}:x_{1},...,x_{k}\geq0}
\mathbb{E}\left[f\left(\bar{Y},\bar{Z}\right)\middle|\ \bar{Y}_{\sum}
=\bar{x}\right]\Pr\left(\bar{Y}_{\sum}=\bar{x}\right)\\
&\geq \mathbb{E}\left[f\left(\bar{Y},\bar{Z}\right)\middle|\ \bar{Y}_{\sum}=\bar{\lambda}\right]\Pr\left(\bar{Y}_{\sum}=\bar{\lambda}\right)
=\mathbb{E}\left[f\left(\bar{X},\bar{Z}\right)\right]
\Pr\left(\bar{Y}_{\sum}=\bar{\lambda}\right)\\
&= \mathbb{E}\left[f\left(\bar{X},\bar{Z}\right)\right]
\prod_{i}\frac{\colballs{i}^{\colballs{i}}}{\colballs{i}!}e^{-\colballs{i}}
\geq\mathbb{E}\left[f\left(\bar{X},\bar{Z}\right)\right]
\frac{e^{-k}}{\sqrt{\prod_{i}\colballs{i}}},
\end{align*}
where, in the last inequality, we use that, by Stirling's approximation,
$a!\leq e\sqrt{a}(\frac{a}{e})^{a}$ for any $a>0$.
\end{proof}

From \lemref{pois-cond} and \lemref{pois-apx}, we get the following
general result which says that if a generic event $\eventhighprob$ holds w.h.p
in \ppapx, it also holds w.h.p. in \porig. 

\begin{lem}
\label{lem:pois-apx-prob}
    Given the \config and the number of active nodes at
    the beginning of a fixed phase $j$, let ${\eventhighprob}$ be an event
    that, at the end of that phase, holds with probability at least
    $1-n^{-\constpoisapx}$ in \ppapx, for some $\constpoisapx> ( k \log h ) / (
    2 \log n ) $ with $h = \sum_i h_i$.\footnote{
        Note that, if \N[j] is not yet fixed, the parameters
        $\colballs{i}$ of \ppapx associated to phase $j$ are random variables.
        However, if the \config and the number of active nodes at the beginning
        of phase $j$ are given, then $h = \sum_i h_i = |\N| = |\M|$ is fixed.}
    Then, at the end of phase $j$, ${\eventhighprob}$ holds
    w.h.p. also in \porig.
\end{lem}
\begin{proof}
    Thanks to \claimref{prot_eq_bib}, it suffices to prove that, at the end of
    phase $j$, ${\eventhighprob}$ holds w.h.p. in \pbib.

    Let $\ceventhighprob$ be the complementary event of $\eventhighprob$. Let $
    \totballs = |\M| $ be the number of balls that are thrown in \pbib
    associated to phase $j$, where 
    $\colballs{i}$ balls have color $i$ with $i\in\left\{ 1,...,k\right\} $ 
    and $\sum_{i}\colballs{i}=\totballs$. Let 
    $\left\{ X_{u,i}\right\}_{\substack{u\in\left\{ 1,...,n\right\}, 
    i\in\left\{ 1,...,k\right\} } }$
    be the number of $i$-colored balls that end up in bin $u$, let 
    $\left\{ Y_{u,i}\right\}_{\substack{u\in\left\{ 1,...,n\right\}, 
    i\in\left\{ 1,...,k\right\} } }$ 
    be the independent r.v. of \ppapx such that
    $Y_{u,i}\sim$Poisson$(\colballs{i}/n)$
    and let $Z_{1},...,Z_{n}$ be integer valued r.v. independent from
    the $X_{u,i}$s and $Y_{u,i}$s. 
    
    Fix any realization of \N[j], i.e. any re-coloring of the balls in the
    first step of \pbib. By choosing $f$ in \lemref{pois-apx} as the binary
    r.v. indicating whether event $\ceventhighprob$ has occurred, where
    $\ceventhighprob$ is a function of the r.v.
    $X_{1,1},...,X_{n,1},$ $X_{n,2},...,X_{n,n},$ $Z_{1},...,Z_{n}$,
    we get
    \begin{align}
    \Pr\left( \ceventhighprob\left(X_{1,1},...,X_{n,n},Z_{1},...,Z_{n} 
    \right) \,\middle|\, \N[j]\right) 
    \leq e^{k}\sqrt{\prod_{i}{\colballs{i}}}\,
    \Pr\left( \ceventhighprob\left(Y_{1,1},...,Y_{n,n},Z_{1},...,Z_{n} 
    \right) \,\middle|\, \N[j]\right).
    \label{eq:pois-apx}
    \end{align}

    Thus, from Eq.~(\ref{eq:pois-apx}), the Inequality of arithmetic and
    geometric means and the hypotheses on the probability of $\eventhighprob$,
    we get 
\begin{align*}
	\Pr\left(\ceventhighprob\left(X_{1,1},...,X_{n,n},
    Z_{1},...,Z_{n}\right) \,\middle|\, \N[j]\right) 
    &\leq e^{k}\sqrt{\prod_{i}{\colballs{i}}}\,
	\Pr\left(\ceventhighprob\left(Y_{1,1},...,
    Y_{n,n},Z_{1},...,Z_{n}\right) \,\middle|\, \N[j]\right)\\
	&\leq e^{k} \left( \frac hk \right)^{\frac k2}
	\Pr\left(\ceventhighprob\left(Y_{1,1},...,
    Y_{n,n},Z_{1},...,Z_{n}\right) \,\middle|\, \N[j]\right)
\end{align*}
Finally, let $\realizsetN$ be the set of all possible realizations of \N[j].
By the law of total probability over $\realizsetN$, we get that 
\begin{align*}
    &\sum_{s \in \realizsetN}
        \Pr\left(\ceventhighprob\left(X_{1,1},...,X_{n,n},
        Z_{1},...,Z_{n}\right) \,\middle|\, \N[j] = s \right) 
        \Pr \left( \N[j] = s \right)\\
    &\leq e^{k} \left( \frac hk \right)^{\frac k2}
        \sum_{s \in \realizsetN} 
        \Pr\left(\ceventhighprob\left(Y_{1,1},...,
        Y_{n,n},Z_{1},...,Z_{n}\right) \,\middle|\, \N[j] = s \right) 
        \Pr \left( \N[j] = s \right)\\
    &\leq e^{k} \left( \frac hk \right)^{\frac k2}
        \Pr\left(\ceventhighprob\left(Y_{1,1},...,
        Y_{n,n},Z_{1},...,Z_{n}\right) \right) \\
    &\leq \frac{ e^{k} }{ k^{ \frac k2 } } \ h^{\frac k2} 
    {n^{-\constpoisapx}}\leq n^{-\Theta(1)},
\end{align*}
where in the first inequality of the last line we used the hypotheses on the
probability of $\ceventhighprob$. 
\end{proof}

We now analyze the two stages of our protocol, starting with Stage~1. 
Note that, in the following two sections, the statements about the evolution of
the process refer to \porig.

\subsection{Stage 1\label{sub:Stage-1}}

The rule of Stage 1 is aimed at guaranteeing that, w.h.p., the system reaches
a target \config from which the rumor-spreading problem
becomes an instance of the plurality consensus problem. More precisely,
we have the following.

\begin{lem}
\label{lem:stage1}
Stage 1 takes $O(\frac{1}{\epsilon^{2}}\log
n)$ rounds, after which w.h.p. all nodes are active
and $\mathbf{c}^{(\tau_{T+1})}$ is
$\delta$-biased toward the correct opinion, with $\delta=\Omega(\sqrt{\log
n/n})$.
\end{lem}

\begin{proof}
The fact that an undecided node becomes opinionated during a phase only depends
on whether it gets a message during that phase, regardless of the value of such
messages. Hence, the proof that, w.h.p., $a^{(\tau_{T+1})}=1$ is reduced to the
analysis of the rule of Stage 1 as an information spreading process.  First, by
carefully exploiting the Chernoff bound and \lemref{pois-apx-prob}, we can
establish \claimref{bootstrap} and \claimref{opinion-growth} below: 

\begin{claim}
\label{claim:bootstrap}
W.h.p., at the end of phase 0, we have 
$
    \nfrac s{\epsilon^{2}} \nfrac {\log n}{3n} \leq a^{(\tau_{0})}\leq \nfrac s{\epsilon^{2}} \nfrac{\log n} n
$.
\end{claim}
\smallskip

\begin{claim}
\label{claim:opinion-growth}
W.h.p., at the end of phase $\pc$, $1\leq\pc\leq T$, we have 
\[
    (\beta/\epsilon^{2}+1)^{\pc}a^{(\tau_{0})}/8\leq a^{(\tau_{\pc})}\leq(\beta/\epsilon^{2}+1)^{\pc}a^{(\tau_{0})}.
\]
\end{claim}
\begin{proof}[Proof of \claimref{bootstrap} and \claimref{opinion-growth}]
The probability that, in the \porig, an undecided node becomes opinionated at the
end of phase $j$ is $1-(1-\frac{1}{n})^{\totballs}$ where $\totballs$ is the
number of messages sent during that phase. In \ppapx, this probability  is
$1-e^{-\frac{\totballs}{n}}$. By using that $e^{\frac{x}{1+x}}\leq1+x\leq
e^{x}$ for $\left|x\right|<1$ we see that $1-e^{-\frac{\totballs}{n}}\leq1 -(
1-{1}/{n})^{\totballs}\leq1-e^{-\frac{\totballs}{n-1}}$. Thus, we can prove
\claimref{bootstrap} and \claimref{opinion-growth} for \ppapx by repeating 
essentially the same calculations as in the proofs of Claim 2.2 and 2.4
in \cite{FHK15}. Since the Poisson distributions
in \ppapx are independent, we can apply the Chernoff bound as claimed in 
\cite{FHK15}. Finally, we can prove that the 
statements hold also for \porig thanks to Lemma \ref{lem:pois-apx}. 
\end{proof}

From the previous two claims, and by the definition of $T$ we get the following.

\begin{lem}
\label{lem:end-stage1}
W.h.p., at the end of phase $T$, we have $a^{(\tau_{T+1})}=\Omega((\beta/\epsilon^{2}+1)^{T}a^{(\tau_{0})})=\Omega(\epsilon^{2})$.
\end{lem}

Finally, from \lemref{end-stage1}, an application of the Chernoff bound
gives us the following.

\begin{lem}
	\label{lem:end_of_stage1}
	W.h.p., at the end of Stage 1, all nodes are opinionated.
\end{lem}

As for the fact that, w.h.p., $\mathbf{c}^{(\tau+1)}$ is
a $\delta$-biased \config with $\delta=\Omega(\sqrt{\log n/n})$,
we can prove the following. 

\begin{lem}
\label{lem:correct-stage1}
W.h.p., at the end of each phase $\pc$ of
Stage 1, we have an $(\epsilon/2)^{\pc}$-biased
\config.
\end{lem}
\begin{proof}
We prove the lemma by induction on the phase number. 
The case $\pc=1$ is a direct application
of \lemref{cb-diff} to $c_{m}^{(\tau_{1})}-c_{i}^{(\tau_{1})}$ ($i\neq m$),
where the number of opinionated nodes is given by \claimref{bootstrap}, and,
where the independence of the r.v. follows from the fact that each node that
becomes opinionated in the first phase has necessarily received the messages
from the source-node. 
Now, suppose that the lemma holds up to phase $\pc-1\leq T$. 
Let $S_{\pc}=\left\{ u\middle|\ \pc_{u}=\pc\right\}$
be the set of nodes that become opinionated during phase $\pc$. Recall the
definition of $\M[ \pc ]$ and $\N[ \pc ]$ from \subref{push-balls}, and
observe that
$\left|\M[ \pc ]\right|= \left|\N[ \pc ]\right|= \left(\tau_{\pc}-
\tau_{\pc-1}\right)n\cdot a^{(\tau_{\pc-1})}$, and that the
number of times opinion $i$ occurs in $M_{\pc}$ is
$\left|M_{\pc}\right|c_{i}^{(\tau_{\pc-1})}$. Let us
identify each message in $M_{\pc}$ with a distinct number in 
$1,...,\left|M_{\pc}\right|$, and let $\left\{ X_{w}(i)\right\}
_{w\in\left\{ 1,...,\left|M_{\pc}\right|\right\} }$ be the binary r.v.
\st $ X_{w}(i)=1$ if and only if $w$ is $i$ after the action of the
noise. The frequency of opinion $i$ in $N_{\pc}$ is
$\frac{1}{\left|N_{\pc}\right|}\sum_{w=1}^{\left|N_{\pc}\right|}X_{w}(i)$.

Thanks to \lemref{pois-apx-prob}, it suffices to prove the lemma for \ppapx.
By definition, in \ppapx, for each $i$, the number of messages with opinion $i$
that each node receives conditional on $\N[ \pc ]$ follows a
Poisson$(\frac{1}{n}\sum_{w=1}^{\left|\N[\pc]\right|}X_{w}\left(i)\right)$
distribution. Each node $u$ that becomes opinionated during phase $\pc$ gets
at least one message during the phase. Thus, from \lemref{pois-cond}, the
probability that $u$ gets opinion $i$ conditional on $\N[ \pc ]$ is
\[
	\frac{\sum_{w=1}^{\left|\N[\pc]\right|} X_{w}\left(i\right)
}{\sum_{i=1}^{k}\sum_{w=1}^{\left|\N[\pc]\right|} X_{w}\left(i\right)}=
\frac{1} {\left|\N[\pc]\right|}
\sum_{w=1}^{\left|\N[\pc]\right|}X_{w}\left(i\right).
\]
Since opinionated nodes never change opinion during Stage 1, 
the bias of $\mathbf{c}^{(\tau_{\pc})}$
is at least the minimum between the bias of $\mathbf{c}^{(\tau_{\pc-1})}$
and the bias among the newly opinionated nodes in $S_{\pc}$.
Hence, we can apply the Chernoff bound to the nodes in $S_{\pc}$
to prove that the bias at the end of phase $\pc$ is, w.h.p.\footnote{We remark that Eq. (\ref{eq:bias_in_nodes})
concerns the value of $\Pr ( c_{m}^{\left(\tau_{\pc}\right)}
-c_{i}^{\left(\tau_{\pc}\right)} | \N[ \pc ])$, which is a random variable.},
\begin{equation}
\Pr \left( c_{m}^{\left(\tau_{\pc}\right)}
-c_{i}^{\left(\tau_{\pc}\right)} \middle| \N[ \pc ]\right)
\geq\left(\frac{1}{\left|\N[\pc]\right|}
\sum_{w=1}^{\left|\N[\pc]\right|}
X_{w}\left(m\right)-\frac{1}{\left|\N[\pc]\right|}
\sum_{w=1}^{\left|\N[\pc]\right|} X_{w} \left(i\right)\right)
\left(1-\tilde{\delta}_{\pc}\right),
 \label{eq:bias_in_nodes}
\end{equation}
where $\tilde{\delta}_{\pc}=O(\sqrt{\log n/|S_{\pc}|})$.

Moreover, note that 
\[
\mathbb{E}\left[\frac{1}{\left|N_{\pc}\right|}
	\sum_{w=1}^{\left|N_{\pc}\right|}X_{w}\left(i\right)
	\middle|\ \mathbf{c}^{\left(\tau_{\pc-1}\right)},
	a^{\left(\tau_{\pc-1}\right)}\right]=
	\left(\mathbf{c}^{\left(\tau_{\pc-1}\right)}\cdot P\right)_{i}.
\]
Furthermore, (conditional on $\mathbf{c}^{(\tau_{\pc-1})}$ and
$a^{(\tau_{\pc-1})}$) the r.v. $\left\{ X_{w} (i)
\right\}_{w \in \left\{ 1,...,\left|N_{\pc}\right| \right\} }$ are
independent.  
Thus, for each $i\neq m$, from \claimref{opinion-growth},
and by applying the Chernoff bound on $\sum_{w=1}^{\left|N_{\pc}\right|}X_{w}(m)$,
and on $\sum_{w=1}^{\left|N_{\pc}\right|}X_{w}(i)$, we
get that w.h.p. 
\begin{equation}
	\frac{1}{\left|N_{\pc}\right|} \sum_{w=1}^{\left|N_{\pc}\right|}
	X_{w}\left(m\right)-\frac{1}{\left|N_{\pc}\right|}
	\sum_{w=1}^{\left|N_{\pc}\right|} X_{w}\left(i\right)
	\geq\left(1-\delta_{\pc}\right)2^{-\pc+1}\epsilon^{\pc},
	\label{eq:bias_in_noise}
\end{equation}
where $\delta_{\pc}=O(\sqrt{\log n/|N_{\pc}|})$.

From \claimref{bootstrap} and \claimref{opinion-growth}, it follows that
 $\tilde{\delta}_{\pc},\delta_{\pc}\leq\frac{1}{4}$ w.h.p. Thus by putting
together Eq.~(\ref{eq:bias_in_nodes}) and (\ref{eq:bias_in_noise}) via the chain rule, we
get that, w.h.p., 
\begin{equation*}
c_{m}^{\left(\tau_{\pc}\right)} -c_{i}^{\left(\tau_{\pc}\right)}
\geq\left(1-\tilde{\delta}_{\pc}\right) \left(1-\delta_{\pc}\right)
2^{-\pc+1}\epsilon^{\pc} \geq\left(\frac{\epsilon}{2}\right)^{\pc}.
\label{eq:new-bias}
\end{equation*}

\end{proof}

\lemref{correct-stage1} implies that, w.h.p., we get a bias
$\epsilon^{T+2}=\Omega(\sqrt{\log n/n})$ at the end of Stage 1, which completes
the proof of \lemref{stage1}. 
\end{proof}

\subsection{Stage 2\label{sub:Stage-2}}

As proved in the previous section,  w.h.p., all nodes are
opinionated at the end of Stage~1, and the final \config is $\Omega(\sqrt{\log
n/n})$-biased. Now, we have that the rumor-spreading problem is reduced to an
instance of the plurality consensus problem. The purpose of Stage~2 is to
progressively amplify the initial bias until all nodes support the plurality
opinion, i.e. the opinion originally held by the source node. 

\global\long\def\majority{\operatorname{maj}}
\global\long\def\maj#1#2{\majority_{#2}(#1)}
During the first $T^{\prime}$ phases, it is not hard to see that, by taking
$\constfour$ large enough, a fraction arbitrarily close to 1 of the nodes
receives at least $\ell$ messages, w.h.p. Each node $u$ in such fraction
changes its opinion at the end of the phase. With a slight abuse of notation,
let $\maj u{\ell} = \majority{(\Sam{u})}$ be $u$'s new opinion based on the
$\ell = |\Sam{u}|$ randomly sampled received messages. We show that, w.h.p.,
these new opinions increase the bias of the \config toward the plurality
opinion by a constant factor $>1$. 

For the sake of simplicity, we assume that $\ell$ is odd (see \apxref{Removing-the-Parity}
for details on how to remove this assumption).

\begin{prop}
\label{prop:maj-ampl}Suppose that, at the beginning of phase $j$
of Stage 2 with $0\leq j\leq T^{\prime}-1$, the \config is
$\delta$-biased toward $m$. In \ppapx, if a node
$u$ changes its opinion at the end of the phase, then, for any $i\neq m$,
we have
\begin{equation}
    \Pr\left(\maj u\ell=m\right)-\Pr\left(\maj u\ell=i\right) 
    \geq\sqrt{\frac{2\ell}{\pi}}\frac{g(\delta,\ell)}{e^{(k-2)\ln4}},
    \label{eq:maj_ampl}
\end{equation}
where 
\[
g\left(\delta,\ell\right)=
    \begin{cases}
    \delta(1-\delta^{2})^{\frac{\ell-1}{2}} & \mbox{ if \ensuremath{\delta}<\ensuremath{\frac{1}{\sqrt{\ell}}},}\\
    \sqrt{1/\ell} \; (1-\sqrt{1/\ell})^{\frac{\ell-1}{2}} & \mbox{ if \ensuremath{\delta\geq \frac{1}{\sqrt{\ell}}}.}
    \end{cases}
\]
\end{prop}

First, we prove Eq.~(\ref{eq:maj_ampl}) for $k=2$. We then obtain the
general case by induction. The proof for $k=2$ is based on a known relation between the cumulative
distribution function of the binomial distribution, and the cumulative
distribution function of the beta distribution. This relation is given
by the following lemma.

\begin{lem}
\label{lem:integral_approx}Given $p\in(0,1)$ and $0\leq j\leq \ell$
it holds 
\begin{equation*}
\sum_{j<i\leq \ell}{\ell \choose i}p^{i}\left(1-p\right)^{\ell-i}={\ell \choose j+1}\left(j+1\right)\int_{0}^{p}z^{j}\left(1-z\right)^{\ell-j-1}dz.
\label{eq:bin-beta_relation}
\end{equation*}
\end{lem}
\begin{proof}
By integrating by parts, for $j<\ell-1$ we have 
\begin{align}
{\ell \choose j+1}\left(j+1\right)\int_{0}^{p}z^{j}\left(1-z\right)^{\ell-j-1}dz 
&= {\ell \choose j+1}p^{j+1}\left(1-p\right)^{\ell-j-1} \nonumber \nonumber\\ 
&- {\ell \choose
j+1}\left(\ell-j-1\right)\int_{0}^{p}z^{j+1}\left(1-z\right)^{\ell-j-2}dz 
\nonumber \\ 
&= {\ell \choose j+1}p^{j+1}\left(1-p\right)^{\ell-j-1} \nonumber \nonumber\\
&- {\ell \choose
j+2}\left(j+2\right)\int_{0}^{p}z^{j+1}\left(1-z\right)^{\ell-j-2}dz,
\label{eq:unrolling_binomial-beta}
\end{align}
where, in the last equality, we used the identity
\[
{\ell \choose j}\left(\ell-j\right)={\ell \choose j+1}\left(j+1\right).
\]

Note that when $j=\ell-1$, Eq.~(\ref{eq:maj_ampl}) becomes 
\[
p^{\ell}=\ell\int_{0}^{p}z^{\ell-1}dz.
\]

Hence, we can unroll the recurrence given by Eq.~(\ref{eq:unrolling_binomial-beta})
to obtain
\begin{align*}
{\ell \choose j+1}\left(j+1\right)\int_{0}^{p}z^{j}\left(1-z\right)^{\ell-j-1}dz
&= \sum_{j<i\leq \ell-1}{\ell \choose i}p^{i}\left(1-p\right)^{\ell-i} 
+ \ell\int_{0}^{p}z^{\ell-1}dz\\
&= \sum_{j<i\leq \ell}{\ell \choose i}p^{i}\left(1-p\right)^{\ell-i}
\end{align*}
concluding the proof.
\end{proof}

\lemref{integral_approx} allows us to express the survival function 
of a binomial sample as an integral. Thanks to it, we can prove \propref{maj-ampl}
when $k=2$.
\begin{lem}
\label{lem:bin_maj_ampl}Let $\mathbf{c}=(c_{1},c_{2})$
be a $\delta$-biased \config during Stage 2. In \ppapx, for any node $u$, we have
$
\Pr\left(\maj u\ell=m\right) -\Pr\left(\maj u\ell=3-m\right)
\geq \sqrt{{2\ell} / {\pi}} \cdot g\left(\delta,\ell\right).
$
\end{lem}
\begin{proof}
Without loss of generality, let $m=1$. 
Let $X_{1}^{(\ell)}$ be a r.v. with distribution $Bin(\ell,p_{1})$,
and let $X_{2}^{(\ell)}=\ell-X_{1}^{(\ell)}$. By using
\lemref{integral_approx}, we get 
\begin{align*}
\Pr\left(\maj u\ell=1\right)-\Pr\left(\maj u\ell=2\right)
&= \Pr\left(X_{1}^{\left(\ell \right)}>X_{2}^{\left( \ell \right)}\right)
-\Pr\left(X_{2}^{\left(\ell \right)}>X_{1}^{\left(\ell \right)}\right)\\
&= \sum_{\left\lceil \frac{\ell }{2}\right\rceil 
\leq i\leq \ell }{\ell  \choose i}p_{1}^{i}p_{2}^{\ell -i}
-\sum_{\left\lceil \frac{\ell }{2}\right\rceil 
\leq i\leq \ell }{\ell  \choose i}p_{1}^{\ell -i}p_{2}^{i}\\
&= \sum_{\left\lceil \frac{\ell }{2}\right\rceil 
\leq i\leq \ell }{\ell  \choose i}p_{1}^{i}\left(1-p_{1}\right)^{\ell -i}
-\sum_{\left\lceil \frac{\ell }{2}\right\rceil 
\leq i\leq \ell }{\ell  \choose i}p_{1}^{\ell -i}\left(1-p_{1}\right)^{i}\\
&= {\ell  \choose \left\lceil \frac{\ell }{2}\right\rceil }\left\lceil 
\frac{\ell }{2}\right\rceil \left(\int_{0}^{p_{1}}z^{\left\lfloor 
\frac{\ell }{2}\right\rfloor }\left(1-z\right)^{\left\lfloor 
\frac{\ell }{2}\right\rfloor }dz\right.\\
&- \left. \int_{0}^{p_{2}}z^{\left\lfloor 
\frac{\ell }{2}\right\rfloor }\left(1-z\right)^{\left\lfloor 
\frac{\ell }{2}\right\rfloor }dz\right).
\end{align*}
By setting $t=z-\frac{1}{2}$, and rewriting 
$p_{1} = \frac{p_{1}-p_{2}}{2} + \frac{1}{2}$
and $p_{2}=\frac{p_{2}-p_{1}}{2}+\frac{1}{2}$ we obtain
\begin{align*}
\Pr\left(\maj u\ell=1\right)-\Pr\left(\maj u\ell=2\right)
&= {\ell  \choose \left\lceil \frac{\ell }{2}\right\rceil }\left\lceil 
\frac{\ell }{2}\right\rceil \left(\int_{0}^{p_{1}}z^{\left\lfloor 
\frac{\ell }{2}\right\rfloor }\left(1-z\right)^{\left\lfloor 
\frac{\ell }{2}\right\rfloor }dz \right. \\
&- \left. \int_{0}^{p_{2}}z^{\left\lfloor 
\frac{\ell }{2}\right\rfloor }\left(1-z\right)^{\left\lfloor 
\frac{\ell }{2}\right\rfloor }dz\right)\\
&= {\ell  \choose \left\lceil \frac{\ell }{2}\right\rceil }\left\lceil 
\frac{\ell }{2}\right\rceil \left(\int_{-\frac{1}{2}}^{\frac{p_{1}-p_{2}}{2}}
\left(\frac{1}{4}-t^{2}\right)^{\left\lfloor \frac{\ell }{2}\right\rfloor }dt 
\right.\\
&- \left. \int_{-\frac{1}{2}}^{\frac{-p_{1}-p_{2}}{2}}
\left(\frac{1}{4}-t^{2}\right)^{\left\lfloor 
\frac{\ell }{2}\right\rfloor }dt\right)\\
&= {\ell  \choose \left\lceil \frac{\ell }{2}\right\rceil }\left\lceil 
\frac{\ell }{2}\right\rceil \int_{-\frac{p_{1}-p_{2}}{2}}^{\frac{p_{1}-p_{2}}{2}}
\left(\frac{1}{4}-t^{2}\right)^{\left\lfloor \frac{\ell }{2}\right\rfloor }dt.
\end{align*}
For any $t\in(-\frac{y}{2}, \frac{y}{2})\subseteq(-\frac{p_{1}-p_{2}}{2},  
\frac{p_{1}-p_{2}}{2})$,
it holds 
\[
\left(\frac{1}{4}-t^{2}\right)^{\left\lfloor \frac{\ell }{2}\right\rfloor }
    \geq \left(\frac{1-y^{2}}{4}\right)^{\left\lfloor 
    \frac{\ell }{2}\right\rfloor } .
\]
Thus, for any $y\in(-p_{1}+p_{2},p_{1}-p_{2})$ we have
\begin{equation}
	\int_{-\frac{p_{1}-p_{2}}{2}}^{\frac{p_{1}-p_{2}}{2}}
	\left(\frac{1}{4}-t^{2}\right)^{\left\lfloor \frac{\ell }{2}\right\rfloor }dt
	\geq y\left(\frac{1-y^{2}}{4}\right)^{\left\lfloor 
	\frac{\ell }{2}\right\rfloor }.
\label{eq:int_apx}
\end{equation}
The r.h.s. of Eq.~(\ref{eq:int_apx}) is maximized w.r.t.
$y\in(-p_{1}+p_{2},p_{1}-p_{2})$ when
\[
	y=\min\left\{ p_{1}-p_{2},\frac{1}{\sqrt{2\left\lfloor
	\frac{\ell }{2}\right\rfloor +1}}\right\} =\min\left\{
	p_{1}-p_{2},\frac{1}{\sqrt{\ell }}\right\}.
\]
Hence, for $p_{1}-p_{2}<\frac{1}{\sqrt{\ell }}$, we get
\begin{align*}
\int_{-\frac{p_{1}-p_{2}}{2}}^{\frac{p_{1}-p_{2}}{2}}
\left(\frac{1}{4}-t^{2}\right)^{\left\lfloor \frac{\ell }{2}\right\rfloor }dt
&\geq \left(p_{1}-p_{2}\right)
\left(\frac{1-\left(p_{1}-p_{2}\right)^{2}}{4}\right)^{\left\lfloor 
	\frac{\ell }{2}\right\rfloor }\\
&= 2^{-\ell +1}\left(p_{1}-p_{2}\right)
\left(1-\left(p_{1}-p_{2}\right)^{2}\right)^{\frac{\ell -1}{2}} \\
&= 2^{-\ell +1}g\left(p_{1}-p_{2},\ell \right).
\end{align*}
For $p_{1}-p_{2}\geq\frac{1}{\sqrt{\ell }}$ we get 
\begin{equation*}
\int_{-\frac{p_{1}-p_{2}}{2}}^{\frac{p_{1}-p_{2}}{2}}
\left(\frac{1}{4}-t^{2}\right)^{\left\lfloor \frac{\ell }{2}\right\rfloor }dt
\geq\frac{2^{-\ell +1}}{\sqrt{\ell }}\left(1-\frac{1}{\ell }\right)^{\frac{\ell -1}{2}}
=2^{-\ell +1}g\left(p_{1}-p_{2},\ell \right).
\end{equation*}
By using the fact that $g$ is a non-decreasing function w.r.t. its first
argument, we obtain
\begin{align*}
\Pr\left(\maj u\ell=1\right)-\Pr\left(\maj u\ell=2\right)
&= {\ell  \choose \left\lceil \frac{\ell }{2}\right\rceil }\left\lceil 
\frac{\ell }{2}\right\rceil \int_{-\frac{p_{1}-p_{2}}{2}}^{\frac{p_{1}-p_{2}}{2}}
\left(\frac{1}{4}-t^{2}\right)^{\left\lfloor \frac{\ell }{2}\right\rfloor }dt\\
&\geq {\ell  \choose \left\lceil \frac{\ell }{2}\right\rceil }\left\lceil 
\frac{\ell }{2}\right\rceil 2^{-\ell +1}g\left(p_{1}-p_{2},\ell \right) \\
&\geq {\ell  \choose \left\lceil \frac{\ell }{2}\right\rceil }\left\lceil 
\frac{\ell }{2}\right\rceil 2^{-\ell +1}g\left(\delta,\ell \right).
\end{align*}
Finally, by using the bounds ${2r \choose r} \geq \frac{2^{2r}}{\sqrt{\pi r}}
e^{\frac{1}{9r}}$
(see \lemref{bin_coef_apx}), and $e^{x}\geq1-x$ together with the
identity\footnote{Recall that we are assuming that $\ell $ is odd.}
\[
{\ell  \choose \left\lceil \frac{\ell }{2}\right\rceil }\left\lceil \frac{\ell }{2}\right\rceil ={\ell  \choose \frac{\ell +1}{2}}\frac{\ell +1}{2}={\ell -1 \choose \frac{\ell -1}{2}}\ell ,
\]
we get 
\begin{align*}
\Pr\left(\maj u\ell=1\right)-\Pr\left(\maj u\ell=2\right)
&\geq {\ell  \choose \left\lceil \frac{\ell }{2}\right\rceil }\left\lceil 
\frac{\ell }{2}\right\rceil 2^{-\ell +1}g\left(\delta,\ell \right) \\
&\geq \frac{2^{\ell -1}}{\sqrt{\pi\frac{\ell -1}{2}}}
e^{\frac{2}{9\left(\ell -1\right)}}\ell 
\cdot2^{-\ell +1}g\left(\delta,\ell \right)\\
&\geq  \sqrt{\frac{2\ell }{\pi}}\left(1-\frac{2}{9\left(\ell -1\right)}\right)
\left(1-\frac{1}{\ell }\right)^{-\frac{1}{2}}\cdot g\left(\delta,\ell \right) \\
&\geq \sqrt{\frac{2\ell }{\pi}}\cdot g\left(\delta,\ell \right),
\end{align*}
concluding the proof.
\end{proof}

Next we show how to lower bound the above difference 
with a much simpler expression. 

\begin{lem}
\label{lem:maj_ampl_lower}In \ppapx, during Stage~2, for any node $u$, $\Pr(\maj u\ell=m)-\Pr(\maj u\ell=3-m)$ is at least 
$
\Pr(X_{1}^{(\ell)}>X_{2}^{(\ell)},...,X_{k}^{(\ell)})
-\Pr(X_{i}^{(\ell)}>X_{1}^{(\ell)}, ..., X_{i-1}^{(\ell)},
X_{i+1}^{(\ell)},...,X_{k}^{(\ell)}),
\label{eq:lower_without_ties}
$
where $\bar{X}^{(\ell)}=(X_{1}^{(\ell)},...,X_{k}^{(\ell)})$ follows a multinomial
distribution with $\ell$ trials and probability distribution $\mathbf{c}\cdot P$.
\end{lem}
\begin{proof}
Without loss of generality, let $m=1$. 
Let $\mathbf{x}=(x_{1},...,x_{k})$ denote a generic
vector with positive integer entries such that $\sum_{j=1}^{k}x_{j}=\ell $, 
let $W(\mathbf{x})$ be the set of the greatest entries of $\mathbf{x}$, 
and, for $j\in\left\{ 1,i\right\} $, let 
\begin{itemize}
    \item $A_{j}^{(!)} = \left\{ \mathbf{x} \,|\, W(\mathbf{x})=\{ j\} \right\} $,
    \item $A_{j}^{(=)}= \left\{ \mathbf{x} \,|\,1,i\in W(\mathbf{x}) \right\}$,
    \item $A_{1}^{(\neq)}= \left\{ \mathbf{x} \,|\, 1\in W(\mathbf{x}) \wedge
        i\not\in W(\mathbf{x})\wedge\left|W(\mathbf{x})\right|>1\right\}$ and 
    \item $A_{i}^{(\neq)}= \left\{ \mathbf{x} \,|\, i\in W(\mathbf{x})\wedge1
        \not\in W(\mathbf{x}) \wedge\left|W(\mathbf{x})\right|>1\right\}$.
\end{itemize}

It holds 
\begin{align}
\Pr\left(\maj u\ell=j\right) &= \sum_{\mathbf{x}\in A_{j}^{\left(!\right)}}
\Pr\left(\bar{X}^{\left(\ell \right)}
= \mathbf{x}\right)\Pr\left(\maj u\ell=j\middle|\,\bar{X}^{\left(\ell \right)}
=\mathbf{x}\right)\nonumber \\
&+ \sum_{\mathbf{x}\in A_{j}^{\left(=\right)}}\Pr\left(\bar{X}^{\left(\ell \right)}
=\mathbf{x}\right)\Pr\left(\maj u\ell=j\middle|\,\bar{X}^{\left(\ell \right)}
=\mathbf{x}\right)\nonumber \\
&+ \sum_{\mathbf{x}\in A_{j}^{\left(\neq\right)}}\Pr\left(\bar{X}^{\left(\ell \right)}
=\mathbf{x}\right)\Pr\left(\maj u\ell=j\middle|\,\bar{X}^{\left(\ell \right)}
=\mathbf{x}\right)\nonumber \\
&= \sum_{\mathbf{x}\in A_{j}^{\left(!\right)}}\Pr\left(\bar{X}^{\left(\ell \right)}
=\mathbf{x}\right)+\sum_{\mathbf{x}\in A_{j}^{\left(=\right)}}
\frac{\Pr\left(\bar{X}^{\left(\ell \right)}
=\mathbf{x}\right)}{\left|W\left(\mathbf{x}\right)\right|} \nonumber \\
&+ \sum_{\mathbf{x}\in A_{j}^{\left(\neq\right)}}
\frac{\Pr\left(\bar{X}^{\left(\ell \right)}
=\mathbf{x}\right)}{\left|W\left(\mathbf{x}\right)\right|}
\label{eq:maj_prob}
\end{align}

Let 
\[
\sigma \left(\mathbf{x}\right)=\left(x_{i},...,x_{i-1},x_{1},x_{i+1},...,x_{k}\right)
\]
 be the vector function that swaps the entries $x_{1}$ and $x_{i}$
in $\mathbf{x}$. $\sigma $ is clearly a bijection between the sets 
$A_{1}^{(!)}$,$A_{1}^{(=)}$,$A_{1}^{(\neq)}$
and $A_{i}^{(!)}$, $A_{i}^{(=)}$, 
$A_{i}^{(\neq)}$,
respectively, namely 
\begin{equation*}
	\sigma :A_{1}^{\left(!\right)}
		\hookrightarrow\hspace{-8pt}\rightarrow A_{i}^{\left(!\right)},\
	\sigma :A_{1}^{\left(=\right)}
		\hookrightarrow\hspace{-8pt}\rightarrow A_{i}^{\left(=\right)},\
	\sigma :A_{1}^{\left(\neq\right)}
		\hookrightarrow\hspace{-8pt}\rightarrow A_{i}^{\left(\neq\right)} 
\end{equation*}
where $\hookrightarrow\hspace{-8pt}\rightarrow$ denotes a bijection.

Moreover, for all $\mathbf{x}\in A_{j}^{(=)}$, it holds
\[
\Pr\left(\bar{X}^{\left(\ell \right)}=\mathbf{x}\right)
=\Pr\left(\bar{X}^{\left(\ell \right)}
=\sigma \left(\mathbf{x}\right)\right).
\]
Therefore
\begin{equation}
\sum_{\mathbf{x}\in A_{1}^{\left(=\right)}}\Pr\left(\bar{X}^{\left(\ell \right)}
=\mathbf{x}\right)=\sum_{\sigma \left(\mathbf{x}\right)\in A_{1}^{\left(=\right)}}
\Pr\left(\bar{X}^{\left(\ell \right)}
=\sigma \left(\mathbf{x}\right)\right)=\sum_{\mathbf{x}\in A_{i}^{\left(=\right)}}
\Pr\left(\bar{X}^{\left(\ell \right)}
=\mathbf{x}\right).
\label{eq:Aeq_prob}
\end{equation}

Furthermore, for all $\mathbf{x}\in A_{1}^{(\neq)}$, we have 
\begin{align}
\Pr\left(\bar{X}^{\left(\ell \right)}=\mathbf{x}\right)
&= {\ell  \choose x_{1}\ ...\ x_{k}}p_{1}^{x_{1}}
\dotsc p_{i}^{x_{i}}\dotsc p_{k}^{x_{k}} \nonumber \\
&> {\ell  \choose x_{1}\ ...\ x_{k}}p_{1}^{x_{i}}
\dotsc p_{i}^{x_{1}}\dotsc p_{k}^{x_{k}} \nonumber \\
&= \Pr\left(\bar{X}_{1}^{\left(\ell \right)}
=\sigma \left(\mathbf{x}\right)\right),
\label{eq:case_ii_discard}
\end{align}
where $\sigma (\mathbf{x})\in A_{i}^{(\neq)}$. From
Eq.~(\ref{eq:case_ii_discard}) we thus have that 
\begin{equation}
\sum_{\mathbf{x}\in A_{1}^{\left(\neq\right)}}
\Pr\left(\bar{X}^{\left(\ell \right)}
=\mathbf{x}\right)>\sum_{\sigma \left(\mathbf{x}\right)
\in A_{1}^{\left(\neq\right)}}\Pr\left(\bar{X}^{\left(\ell \right)}
=\sigma \left(\mathbf{x}\right)\right)=\sum_{\mathbf{x}
\in A_{i}^{\left(\neq\right)}}\Pr\left(\bar{X}^{\left(\ell \right)}
=\mathbf{x}\right).
\label{eq:Aineq_prob}
\end{equation}

From Eq.~(\ref{eq:maj_prob}), (\ref{eq:Aeq_prob}) and (\ref{eq:Aineq_prob})
we finally get
\begin{align*}
\Pr\left(\maj u\ell=1\right)-\Pr\left(\maj u\ell=i\right)
&= \sum_{\mathbf{x}\in A_{1}^{\left(!\right)}} \Pr\left(\bar{X}^{\left(\ell
\right)}=\mathbf{x}\right)+\sum_{\mathbf{x}\in
A_{1}^{\left(=\right)}}\frac{\Pr\left(\bar{X}^{\left(\ell
\right)}=\mathbf{x}\right)}{\left|W\left(\mathbf{x}\right)\right|} \\
&+ \sum_{\mathbf{x}\in
A_{1}^{\left(\neq\right)}}\frac{\Pr\left(\bar{X}^{\left(l\right)} =
\mathbf{x}\right)}{\left|W\left(\mathbf{x}\right)\right|} - \sum_{\mathbf{x}\in
A_{i}^{\left(!\right)}}\Pr\left(\bar{X}^{\left(\ell\right)}=\mathbf{x}\right) \\
&- \sum_{\mathbf{x}\in
A_{i}^{\left(=\right)}}\frac{\Pr\left(\bar{X}^{\left(\ell\right)} =
\mathbf{x}\right)}{\left|W\left(\mathbf{x}\right)\right|} - \sum_{\mathbf{x}\in
A_{i}^{\left(\neq\right)}}\frac{\Pr\left(\bar{X}^{\left(\ell
\right)}=\mathbf{x}\right)}{\left|W\left(\mathbf{x}\right)\right|}\\
&\geq \sum_{\mathbf{x}\in A_{1}^{\left(!\right)}}\Pr\left(\bar{X}^{\left(\ell
\right)}=\mathbf{x}\right)-\sum_{\mathbf{x}\in
A_{i}^{\left(!\right)}}\Pr\left(\bar{X}^{\left(\ell
\right)}=\mathbf{x}\right)\\
&= \Pr\left( W( \bar{X}^{\left(\ell \right)} ) = \{
X_{1}^{\left(\ell \right)} \} \right) - \Pr\left(  W(
\bar{X}^{\left(\ell \right)} ) = \{ X_{i}^{\left(\ell
\right)}\}\right),
\end{align*}
concluding the proof of Lemma~\ref{lem:maj_ampl_lower}.
\end{proof}

Intuitively, \lemref{maj_ampl_lower}
says that the set of events in which a tie occurs among the most frequent
opinions in the node's sample of observed messages does not favor
the probability that the node picks the wrong opinion. Thus, by avoiding
considering those events, we get a lower bound on $\Pr(\maj u\ell=1)-\Pr(\maj u\ell=i)$.

Thanks to \lemref{maj_ampl_lower}, the proof of Eq.~(\ref{eq:maj_ampl})
reduces to proving the following.

\begin{lem}
	\label{lem:clear_maj_ineq}
	For any fixed $k$, and with $\bar{X}$ defined as in \lemref{maj_ampl_lower}, we have 
	\begin{equation}
		\Pr(X_{1}^{\left(\ell\right)}>X_{2}^{\left(\ell\right)}, ...,
		X_{k}^{\left(\ell\right)})
		-\Pr(X_{i}^{\left(\ell\right)}>X_{1}^{\left(\ell\right)}, ...,
		X_{i-1}^{\left(\ell\right)}, X_{i+1}^{\left(\ell\right)}, ...,
		X_{k}^{\left(\ell\right)})\geq \sqrt{2\ell/\pi} \;
		\frac{g\left(\delta,\ell\right)}{4^{k-2}}.
		\label{eq:maj_ampl_lower}
	\end{equation} 
\end{lem}

\begin{proof}
We prove Eq.~(\ref{eq:maj_ampl_lower}) by induction. \lemref{bin_maj_ampl}
provides us with the base case for $k=2$. Let us assume that, for $k\leq\kappa$,
Eq.~(\ref{eq:maj_ampl_lower}) holds. For $k=\kappa+1$, by using the
law of total probability, we have 
\begin{align}
&\Pr\left(X_{1}^{\left(\ell \right)} > X_{2}^{\left(\ell \right)}, ..., 
X_{\kappa+1}^{\left(\ell \right)}\right) - \Pr\left(X_{i}^{\left(\ell \right)} 
> X_{1}^{\left(\ell \right)}, ..., X_{i-1}^{\left(\ell \right)}, 
X_{i+1}^{\left(\ell\right)},...,X_{\kappa+1}^{\left(\ell\right)}\right)\nonumber\\
&\geq\sum_{h=0}^{\left\lfloor \frac{\ell}{\kappa+1}\right\rfloor }
\Pr\left(X_{1}^{\left(\ell\right)} > X_{2}^{\left(\ell\right)}, ..., 
X_{\kappa+1}^{\left(\ell\right)}\middle|\, 
X_{\kappa+1}^{\left(\ell\right)}=h\right)
\Pr\left(X_{\kappa+1}^{\left(\ell\right)}=h\right)\nonumber\\
&-\sum_{h=0}^{\left\lfloor \frac{\ell}{\kappa+1}\right\rfloor
}\Pr\left(X_{i}^{\left(\ell\right)} >
X_{1}^{\left(\ell\right)},...,X_{i-1}^{\left(\ell\right)},
X_{i+1}^{\left(\ell\right)}, ..., X_{\kappa+1}^{\left(\ell\right)}\middle|\,
X_{\kappa+1}^{\left(\ell\right)} =
h\right)\Pr\left(X_{\kappa+1}^{\left(\ell\right)}=h\right)
\label{eq:tot_prob_for_maj_ampl}.
\end{align}
Now, 
    $\arg\max_j \{ X_{j}^{\left(\ell \right)} \} 
    = X_{i}^{\left(\ell \right)}$ 
and $X_{\kappa+1}^{\left(\ell \right)}\leq\left\lfloor
\frac{\ell }{\kappa+1}\right\rfloor $ together imply
$X_{i}^{\left(\ell \right)}>X_{\kappa+1}^{\left(\ell \right)}$.
Thus, in the r.h.s. of Eq.~(\ref{eq:tot_prob_for_maj_ampl}), we have
\[
\Pr\left(X_{1}^{\left(\ell \right)}>X_{2}^{\left(\ell
    \right)},...,X_{\kappa+1}^{\left(\ell
    \right)}\middle|\,X_{\kappa+1}^{\left(\ell
    \right)}=h\right) 
= \Pr\left(X_{1}^{\left(l\right)}>X_{2}^{\left(\ell
    \right)},...,X_{\kappa}^{\left(\ell
    \right)}\middle|\,X_{\kappa+1}^{\left(\ell
    \right)}=h\right)
\]
and 
\begin{align*}
&\Pr\left(X_{i}^{\left(\ell \right)}>X_{1}^{\left(\ell
\right)},...,X_{i-1}^{\left(\ell \right)},X_{i+1}^{\left(\ell
\right)},...,X_{\kappa+1}^{\left(\ell
\right)}\middle|\,X_{\kappa+1}^{\left(\ell \right)}=h\right)\\
&=\Pr\left(X_{i}^{\left(\ell \right)}>X_{1}^{\left(\ell
\right)},...,X_{i-1}^{\left(\ell \right)},X_{i+1}^{\left(\ell
\right)},...,X_{\kappa}^{\left(\ell \right)}\middle|\,X_{\kappa+1}^{\left(\ell
\right)}=h\right).
\end{align*}

Moreover, $X^{(\ell )}$ follows a multinomial distribution
with parameters $\mathbf{p}$ and $\ell $. Thus $X_{k}^{(\ell )}=h$
implies that the remaining entries $X_{1}^{(\ell )},...,X_{k-1}^{(\ell )}$
follow a multinomial distribution with $l-h$ trials, and distribution
$(\frac{p_{1}}{1-p_{k}}, ..., \frac{p_{k-1}}{1-p_{k}})$. Let
$Y^{(\ell -h)} =(Y_{1}^{(\ell -h)}, ...,
Y_{k-1}^{(\ell -h)})$ be the distribution of
$X_{1}^{(\ell )},...,X_{k-1}^{(\ell )}$ conditional on
$X_{k}^{(\ell )}=h$. From Eq.~(\ref{eq:tot_prob_for_maj_ampl})
we get 
\begin{align}
&\Pr\left(X_{1}^{\left(\ell \right)}>X_{2}^{\left(\ell
\right)},...,X_{\kappa+1}^{\left(\ell
\right)}\right)-\Pr\left(X_{i}^{\left(\ell \right)}>X_{1}^{\left(\ell
\right)},...,X_{i-1}^{\left(\ell \right)},X_{i+1}^{\left(\ell
\right)},...,X_{\kappa+1}^{\left(\ell \right)}\right) \nonumber \\
&\geq\sum_{h=0}^{\left\lfloor \frac{\ell }{\kappa+1}\right\rfloor
}\Pr\left(X_{1}^{\left(\ell \right)}>X_{2}^{\left(\ell
\right)},...,X_{\kappa}^{\left(l\right)}\middle|\,X_{\kappa+1}^{\left(\ell
\right)}=h\right)\Pr\left(X_{\kappa+1}^{\left(\ell \right)}=h\right) \nonumber\\
&-\sum_{h=0}^{\left\lfloor \frac{\ell }{\kappa+1}\right\rfloor
}\Pr\left(X_{i}^{\left(\ell \right)}>X_{1}^{\left(\ell
\right)},...,X_{i-1}^{\left(\ell \right)},X_{i+1}^{\left(\ell
\right)},...,X_{\kappa}^{\left(\ell \right)}\middle|\,X_{\kappa+1}^{\left(\ell
\right)}=h\right)\Pr\left(X_{\kappa+1}^{\left(\ell \right)}=h\right) \nonumber\\
&\geq\sum_{h=0}^{\left\lfloor \frac{\ell }{\kappa+1}\right\rfloor
}\left(\Pr\left(Y_{1}^{\left(\ell -h\right)}>Y_{2}^{\left(\ell
-h\right)},...,Y_{\kappa}^{\left(\ell -h\right)}\right)-\right. \nonumber \\
&\hfill\left.-\Pr\left(Y_{i}^{\left(l-h\right)}>Y_{1}^{\left(l-h\right)},...,
Y_{i-1}^{\left(l-h\right)},Y_{i+1}^{\left(l-h\right)},...,
Y_{\kappa}^{\left(l-h\right)}\right)\right)
\Pr\left(X_{\kappa+1}^{\left(l\right)}=h\right).
\label{eq:inductive_in_maj_ampl}
\end{align}
Now, using the inductive hypothesis on the r.h.s. of 
Eq.~(\ref{eq:inductive_in_maj_ampl}) we get
\begin{align*}
&\sum_{h=0}^{\left\lfloor \frac{\ell }{\kappa+1}\right\rfloor }
\left(\Pr\left(Y_{1}^{\left(\ell -h\right)}>Y_{2}^{\left(\ell
-h\right)},...,Y_{\kappa}^{\left(\ell -h\right)}\right) \right.\\
&\left.-\Pr\left(Y_{i}^{\left(\ell -h\right)}>Y_{1}^{\left(\ell
-h\right)},...,Y_{i-1}^{\left(\ell -h\right)},Y_{i+1}^{\left(\ell
-h\right)},...,Y_{\kappa}^{\left(\ell
-h\right)}\right)\right)\Pr\left(X_{\kappa+1}^{\left(\ell \right)}=h\right)\\
&\geq\sum_{h=0}^{\left\lfloor \frac{\ell }{\kappa+1}\right\rfloor
}\left(\sqrt{\frac{2\ell -2h}{\pi}}\frac{g\left(\delta,\ell
-h\right)}{4^{\kappa-2}}\right)\Pr\left(X_{\kappa+1}^{\left(\ell
\right)}=h\right)\\
&\geq\sqrt{\frac{2\ell }{\pi}}\frac{g\left(\delta,\ell
\right)}{4^{\kappa-2}}\cdot\sum_{h=0}^{\left\lfloor \frac{\ell
}{\kappa+1}\right\rfloor }\sqrt{1-\frac{h}{\ell
}}\Pr\left(X_{\kappa+1}^{\left(\ell \right)}=h\right),
\end{align*}
where, in the last inequality, we used the fact that $g$ is a non-increasing
function w.r.t. the second argument (see \lemref{g_monotone}). 

It remains to show that 
\begin{equation*}
\sum_{h=0}^{\left\lfloor \frac{\ell }{\kappa+1}\right\rfloor }\sqrt{1-\frac{h}{l}}\Pr\left(X_{\kappa+1}^{\left(\ell \right)}=h\right)\geq\frac{1}{4}.
\end{equation*}
Let $W_{\kappa+1}^{(\ell )}$ be a r.v. with probability distribution 
$Bin(\ell ,\frac{1}{\kappa+1})$.
Since $X_{\kappa+1}^{(\ell )}\sim Bin(\ell ,p_{\kappa+1})$
with $p_{\kappa+1}\leq\frac{1}{\kappa+1}$, a standard coupling
argument (see for example \cite[Exercise 1.1.]{dubhashi1998concentration}),
enables to show that 
\[
\Pr\left(X_{\kappa+1}^{\left(\ell \right)}\leq h\right)\geq\Pr\left(W_{\kappa+1}^{\left(\ell \right)}\leq h\right).
\]
 Hence, we can apply the central limit theorem (\lemref{bin_clt})
on $W_{\kappa+1}^{(\ell )}$, and get that, for any 
$\tilde{\epsilon}\leq\frac{2-\sqrt{3}}{4}$,
there exists some fixed constant $\ell _{0}$ such that, 
for $\ell \geq \ell _{0}$, we have
\begin{equation}
\Pr\left(X_{\kappa+1}^{\left(\ell \right)}\leq\frac{\ell }{\kappa+1}\right)\geq\Pr\left(W_{\kappa+1}^{\left(\ell \right)}\leq\frac{\ell }{\kappa+1}\right)\geq\left(\frac{1}{2}-\tilde{\epsilon}\right).
\label{eq:bound_half_bin}
\end{equation}
By using Eq.~(\ref{eq:bound_half_bin}), for $\ell \geq \ell _{0}$ we finally
get that 
\begin{align*}
\sum_{h=0}^{\left\lfloor \frac{\ell }{\kappa+1}\right\rfloor }\sqrt{1-\frac{h}{\ell }}
\Pr\left(X_{\kappa+1}^{\left(\ell \right)}=h\right)
&\geq \sqrt{1-\frac{\left\lfloor \frac{\ell }{\kappa+1}\right\rfloor }{\ell }}
\cdot\sum_{h=0}^{\left\lfloor \frac{\ell }{\kappa+1}\right\rfloor }
\Pr\left(X_{\kappa+1}^{\left(\ell \right)}=h\right)\\
&\geq \sqrt{1-\frac{2}{\kappa+1}}
\cdot\Pr\left(X_{\kappa+1}^{\left(\ell \right)}\leq\frac{\ell }{\kappa+1}\right)\\
&\geq \sqrt{\frac{\kappa-1}{\kappa+1}}
\cdot\left(\frac{1}{2}-\tilde{\epsilon}\right)\geq\sqrt{\frac{1}{3}}
\cdot\left(\frac{1}{2}-\tilde{\epsilon}\right)\geq\frac{1}{4},
\end{align*}
concluding the proof that 
\[
\Pr\left(\maj u\ell=1\right)-\Pr\left(\maj u\ell=i\right)\geq\sqrt{\frac{2\ell }{\pi}}\frac{g\left(\delta,\ell \right)}{e^{\left(k-2\right)\ln4}}.
\]
\end{proof}

By using \propref{maj-ampl}, we can then prove \lemref{final}.

\begin{lem}
\label{lem:final}
W.h.p., at the end of Stage 2, all nodes support
the initial plurality opinion. 
\end{lem}
\begin{proof}
Let $\delta=\Omega(\sqrt{\log n/n})$ be the bias of the
\config at the beginning of a generic phase $\pc<T^{\prime}$
of Stage 2. Thanks to  \propref{maj-ampl}, by choosing the constant $\constfour$
of the phase length large enough, in \ppapx we get that
$ \Pr\left(\maj u\ell=m\right)-\Pr\left(\maj u\ell=i\right)\geq\alpha\delta $
for some constant $\alpha>1$ (provided that $\delta\leq 1/2$). 
Hence, by applying \lemref{cb-diff} in \apxref{techtools} with $\theta=\frac{\alpha}{4}\delta$,
we get 
$ \Pr( c_{m}^{(\tau_{\pc})} -c_{i}^{(\tau_{\pc})} 
\leq{\alpha\delta}/{2}) 
\leq\exp( -{(\alpha\delta)^{2}n}/{16}) 
\leq n^{-\tilde{\alpha}} $
for some constant $\tilde{\alpha}$ that is large enough to apply
\lemref{pois-apx}. Therefore, until $\delta \geq 1/2$, in \ppapx we have that
$ c_{m}^{\left(\tau_{\pc}\right)}
-c_{i}^{\left(\tau_{\pc}\right)} \geq{\alpha}\delta/{2} $
holds w.h.p. From the previous equation it follows that, after
$T^{\prime}$ phases, the protocol has reached an \config with a bias
greater than $1/2$. Thus, by a direct application of \lemref{cb-diff} and
\lemref{pois-apx} to $c_{m}^{(\tau_{T^{\prime}})}
-c_{i}^{(\tau_{T^{\prime}})}$, we get that, w.h.p.,
$c_{m}^{(\tau_{T^{\prime}})}
-c_{i}^{(\tau_{T^{\prime}})}=1$, concluding the proof. 
\end{proof}

Finally, the time efficiency claimed in \thmref{general_rumor_spreading} and
\thmref{general_majority_consensus} directly follows from \lemref{final}, while
the required memory follows from the fact that in each phase each node needs
only to count how many times it has received each opinion, i.e. to count up to
at most $O(\frac{1}{\epsilon^2}\log n)$ w.h.p.

%-------------------------------
\section{On the Notion of $(\epsilon, \delta)$-Majority-Preserving Matrix}
\label{sec:mp_matrices}

In this section we discuss the notion of $(\epsilon, \delta)$-m.p. noise
matrix introduced by Definition \ref{def:noise}. 
Let us consider
Eq.~(\ref{obs:average-intuition-P}). The matrix \textbf{$P$} represents the
``perturbation'' introduced by the noise, and so $(\mathbf{c}\cdot
P)_{m}-(\mathbf{c}\cdot P)_{i}$ measures how much information the system is
losing about the correct opinion $m$, in a single communication round. An
$(\epsilon, \delta)$-m.p. noise matrix is a noise matrix that preserves at
least an $\epsilon$ fraction of bias, provided the initial bias is at least
$\delta$. The $(\epsilon, \delta)$-m.p. property essentially characterizes the
amount of noise beyond which some coordination problems cannot be solved
without further hypotheses on the nodes' knowledge of the matrix $P$. To see
why this is the case, consider an $(\epsilon, \delta)$-m.p. noise matrix for
which there is a {$\delta$-biased \config} $\tilde{\mathbf{c}}$ \st
$\left(\mathbf{\tilde{\mathbf{c}}}\cdot
P\right)_{m}-(\mathbf{\tilde{\mathbf{c}}}\cdot P)_{i}<0$ for some opinion~$i$.
Given \config $\tilde{\mathbf{c}}$, \emph{from each node's perspective, opinion
$m$ does not appear to be the most frequent opinion}. Indeed, the messages that
are received are more likely to be $i$ than $m$. Thus, plurality consensus
cannot be solved from \config $\tilde{\mathbf{c}}$.

Observe that verifying whether a given matrix $P$ is $(\epsilon, \delta)$-m.p.
\wrt opinion $m$ consists in checking whether for each $i\neq m$ the value of
the following linear program is at least $\epsilon \delta$:
\begin{align*}
    \text{maximize }  & (P\cdot \bc)_m-(P\cdot \bc)_i\\
    \text{subject to }& \sum_{j} c_j=1,\\
    \text{and }       & \forall j, \, c_j\geq 0, c_m-c_j-\delta\geq 0.
    \label{eq:mp_linearprog}
\end{align*}

We now provide some negative and positive examples of $(\epsilon, \delta)$-m.p.
noise matrices. First, we note that a natural matrix property such as being
diagonally dominant does not imply that the matrix is $(\epsilon, \delta)$-m.p.
For example, by multiplying the following diagonally dominant matrix by the
$\delta$-biased \config $\bc = (1/2 + \delta, 1/2 - \delta, 0)^{\intercal}$, we
see that it does not even preserve the majority opinion at all when ${\epsilon, \delta} < 1/6$:
\begin{equation*}
    \left(
    \begin{array}{ccc}
        \frac{1}{2}+\epsilon & 0                    & \frac{1}{2}-\epsilon\\
        \frac{1}{2}-\epsilon & \frac{1}{2}+\epsilon & 0 \\
        0                    & \frac{1}{2}-\epsilon & \frac{1}{2}+\epsilon 
    \end{array}
    \right).
    \label{eq:mp_not_d_dominant}
\end{equation*}
On the other hand, the following natural generalization of the noise matrix
 in \cite{FHK14} (see Eq. (\ref{eq:binary_n_m})), is $(\epsilon, \delta)$-m.p. for every $\delta>0$ 
\wrt any opinion: 
\begin{equation*}
    (P)_{i,j} = p_{i,j} = 
    \begin{cases}
        \frac 1k + \epsilon & \text{ if }i=j,\\
        \frac 1k - \frac \epsilon{k-1} & \text{ otherwise.}
    \end{cases}
    \label{eq:natural_generalization}
\end{equation*}
More generally, let $P$ be a noise matrix such that
\begin{equation}
    (P)_{i,j} = 
    \begin{cases}
        p & \text{ if }i=j,\\
        q_l \leq q_{i,j}\leq q_u & \text{ otherwise,}
    \end{cases}
    \label{eq:bounded_class}
\end{equation}
for some positive numbers $p$, $q_u$ and $q_l$. 
Since
\begin{align}
    (P\bc)_m - (P\bc)_i 
        &= p c_m + \sum_{j\neq m} q_{j,m} c_j -  p c_i - \sum_{j\neq i} q_{j,i} c_j\nonumber\\
        &\geq p (c_m - c_i)+ \sum_{j\neq m} q_l c_j  - \sum_{j\neq i} q_u c_j \nonumber\\
        &\geq p (c_m - c_i)+ q_l (1-c_m) - q_u (1-c_i) \nonumber\\
        &\geq p (c_m - c_i)+ q_l - q_l c_m - q_u + q_u c_i \nonumber\\
        &\geq p (c_m - c_i)  -q_u (c_m-c_i) - (q_u - q_l) \nonumber\\
        &\geq (p-q_u) (c_m - c_i) - (q_u - q_l) \nonumber\\
        &\geq (p-q_u) \delta - (q_u - q_l).
    \label{eq:conditions_on_ql_qu}
\end{align}
By defining $\epsilon = (p-q_u)/2$, we get that the last line in
Eq. (\ref{eq:conditions_on_ql_qu}) is greater than $\epsilon \delta$ iff $(p-q_u)
\delta /2\geq (q_u - q_l)$, which gives a sufficient condition for any matrix of
the form given in Eq. (\ref{eq:bounded_class}) for being $(\epsilon, \delta)$-m.p.

\section{Conclusion}
In this paper, we solved the general version of rumor spreading and plurality
consensus in biological systems. That is, we have solved these problems for an
arbitrarily large number $k$ of opinions. We are not aware of realistic
biological contexts in which the number of opinions might be a function of the
number $n$ of individuals. Nevertheless, it could be interesting, at least from
a conceptual point to view, to address rumor spreading and plurality consensus
in a scenario in which the number of opinions varies with~$n$. This appears to
be a technically challenging problem. Indeed, extending the results in the
extended abstract of \cite{FHK15} from 2~opinions to any constant number~$k$ of
opinions already required to use complex tools. Yet, several of these tools do
not apply if~$k$ depends on~$n$. This is typically the case of
\propref{maj-ampl}. We let as an open problem the design of stochastic tools
enabling to handle the scenario where $k=k(n)$.

%%-------------------------------
%\vfill
{\small \paragraph{Acknowledgments.} We thank the anonymous reviewers of an
earlier version of this work for their constructive criticisms and
comments, which were of great help in improving the results and their presentation.}

%-------------------------------

%-----------------------------
\newpage

\appendix
\centerline{\Large APPENDIX}

\section{The Reception of Simultaneous Messages}
\label{apx:simultaneous_messages}

In the \model, it may happen that several agents push a message to
the same node $u$ at the same round. In such cases, the model should specify
whether the node receives all such messages, only one of them or neither of
them. Which choice is better depends on the biological setting that is being
modeled: if the communication between the agents of the system is an auditory
or tactile signal, it could be more realistic to assume that simultaneous
messages to the same node would ``collide'', and the node would not be able to
grasp any of them. If, on the other hand, the messages represent visual or
chemical signals (see e.g. \cite{FishConsensus, HouseHunt, Dolev, BCNPST13}),
then it may be unrealistic to assume that nodes cannot receive more than one of
such messages at the same round and besides, by a standard balls-into-bins
argument (e.g. by applying Lemma \ref{lem:pois-apx-prob}), it follows that in
the \model at each round no node receives more than $\bigo(\log n)$
messages w.h.p. In this work we thus consider the model in which all messages
are received, also because such assumption allows us to obtain simpler proofs
than the other variants. We finally note that our protocol does not strictly
need such assumption, since it only requires the nodes to collect a small
random sample of the received messages. However, since we look at the latter
feature as a consequence of active choices of the nodes rather than some
inherent property of the environment, we avoid to weaken the model to the point
that it matches the requirements of the protocol.

\section{Technical Tools}\label{apx:techtools}
\begin{lem}
\label{lem:bin_coef_apx} 
For any integer $r\geq 1$ it holds
\[
	\frac{2^{2r}}{\sqrt{\pi r}} e^{\frac{1}{9r}}\leq{2r \choose r}
	\leq\frac{2^{2r}}{\sqrt{\pi r}} e^{\frac{1}{8r}}.
\]
\end{lem}
\begin{proof}
By using Stirling's approximation \cite{robbins1955remark}
\[
\sqrt{2\pi r}\left(\frac{r}{e}\right)^{r} e^{\frac{1}{12r+1}}
\leq r!\leq\sqrt{2\pi r}\left(\frac{r}{e}\right)^{r} e^{\frac{1}{12r}},
\]
we have 
\begin{equation*}
{2r \choose r}=\frac{\left(2r\right)!}{\left(r!\right)^{2}}
\geq\frac{\sqrt{2\pi2r}\left(\frac{2r}{e}\right)^{2r} 
e^{\frac{1}{12r+1}}}{\left(\sqrt{2\pi r}\left(\frac{r}{e}\right)^{r} 
e^{\frac{1}{12r}}\right)^{2}}=\frac{\sqrt{4\pi r}\left(\frac{2r}{e}\right)^{2r} 
e^{\frac{2}{12r+1}}}{2\pi r\left(\frac{r}{e}\right)^{2r} e^{\frac{1}{24r}}}
=\frac{2^{2r}}{\sqrt{\pi r}} e^{\frac{2}{12r+1}-\frac{1}{24r}}
\geq\frac{2^{2r}}{\sqrt{\pi r}} e^{\frac{1}{9r}}.
\end{equation*}
The proof of the upper bound is analogous (swap $e^{\frac{1}{12r+1}}$
and $e^{\frac{1}{12r}}$ in the first inequality).
\end{proof}

\begin{lem}
\label{lem:bin_clt}
Let $X_{1},...,X_{n}$ be a random sample from
a Bernoulli$(p)$ distribution with $p\in(0,1)$
constant, and let $Z\sim N(0,1)$. It holds
\[
\lim_{n\rightarrow\infty} \sup_{z\in\mathbb{R}}
\left|\Pr\left(\frac{\sum_{i=1}^{n}X_{i}-pn} {\sqrt{n}}\leq
z\right)-\Pr\left(Z\leq\frac{z}{\sqrt{p\left(1-p\right)}}\right)\right|=0.
\]
\end{lem}

\begin{lem}
\label{lem:g_monotone} The function 
\[
g\left(x,y\right)=\begin{cases}
x\left(1-x^{2}\right)^{\frac{y-1}{2}} & \mbox{ if \ensuremath{x<\frac{1}{\sqrt{y}}},}\\
\frac{1}{\sqrt{y}}\left(1-\frac{1}{y}\right)^{\frac{y-1}{2}} & \mbox{ if \ensuremath{x\geq\frac{1}{\sqrt{y}}},}
\end{cases}
\]
with $x\in\left[0,1\right]$ and $y\in\left[1,+\infty\right)$ is
non-decreasing w.r.t. $x$ and non-increasing w.r.t. $y$.\end{lem}
\begin{proof}
To show that $g(x,y)$ is non-decreasing w.r.t. $x$, observe
that 
\[
    \frac {\partial }{ \partial{x} } g\left(x,y\right)=\left(\left(1-x^{2}\right)^{\frac{y-1}{2}}-2x^{2}\left(\frac{y-1}{2}\right)\left(1-x^{2}\right)^{\frac{y-1}{2}}\right)
\]
for $x < y^{ - \frac 12} < 1$, and
\[
\left(1-x^{2}\right)^{\frac{y-1}{2}}-2x^{2}\left(\frac{y-1}{2}\right)\left(1-x^{2}\right)^{\frac{y-1}{2}}\geq0
\]
 for $x < y^{ - \frac 12}$.

To show that $g(x,y)$ is non-increasing w.r.t. $y$, observe
that this is true for $x < y^{ - \frac 12}$.
For $x \geq y^{ - \frac 12}$,
since 
\[
    \frac {\partial}{\partial{y}} \left(\log y^{-\frac{1}{2}}
	+\frac{y-1}{2}\log\left(1-\frac{1}{y}\right)\right)
	=\frac {\partial}{\partial{y}}\left(\frac{y-1}{2}\log\left(y-1\right)-\frac{y}{2}\log
	y\right)\leq0,
\]
we have
\begin{align*}
	\frac {\partial}{\partial{y}} g\left(x,y\right) & 
    =\frac {\partial}{\partial{y}}\exp\left\{ \log
	y^{-\frac{1}{2}}+\frac{y-1}{2}\log\left(1-\frac{1}{y}\right)\right\}
	\leq0,
\end{align*}
concluding the proof.
\end{proof}

\begin{lem}
\label{lem:cb-diff}
Let $\left\{ X_{t}\right\} {}_{t\in[n]}$ be $n$ i.i.d. random variables
such that 
\[
	X_{t}=
	\begin{cases}
		1 & \mbox{ with probability \ensuremath{p},}\\
		0 & \mbox{ with probability \ensuremath{r},}\\
		-1 & \mbox{ with probability \ensuremath{q}.}
	\end{cases}
\]
with $p+r+q=1$.
It holds 
\[
\Pr\left(\sum_{i}X_{t}\leq\left(1-\theta\right)\mathbb{E}\left[\sum_{i}X_{t}\right]-\theta n\right)
\leq\exp\left(-\frac{\theta^{2}}{4}\left(\mathbb{E}\left[\sum_{i}X_{t}\right]+n\right)\right).
\]
\end{lem}
\begin{proof}
Let us define the r.v. 
\begin{equation}
Y_{t}=\frac{X_{t}+1}{2}.
\label{eq:subs_cb-diff}
\end{equation}
We can apply the Chernoff-Hoeffding bound to $Y_{t}$ (see Theorem
1.1 in \cite{dubhashi1998concentration}), obtaining 
\[
\Pr\left(\sum_{i}Y_{i}\leq\left(1-\theta\right)\mathbb{E}\left[\sum_{i}Y_{i}\right]\right)\leq\exp\left(-\frac{\theta^{2}}{2}\mathbb{E}\left[\sum_{i}Y_{i}\right]\right)
\]
for any $\theta\in(0,1)$. Substituting Eq.~(\ref{eq:subs_cb-diff})
we have 
\begin{align*}
\Pr\left(\sum_{i}X_{t}+n\leq\left(1-\theta\right)
\left(\mathbb{E}\left[\sum_{i}X_{t}\right]+n\right)\right)
&= \Pr\left(\sum_{i}X_{t}\leq\left(1-\theta\right)
\mathbb{E}\left[\sum_{i}X_{t}\right]-\theta n\right)\\
&\leq \exp\left(-\frac{\theta^{2}}{4}\left(\mathbb{E}
\left[\sum_{i}X_{t}\right]+n\right)\right),
\end{align*}
concluding the proof.
\end{proof}

\section{Removing the Parity Assumption on $\ell$\label{apx:Removing-the-Parity}}

The next lemma shows that, for $k=2$, the increment of bias at the end
of each phase of Stage 2 in the \ppapx is non-decreasing
in the value of $\ell$, regardless of its parity. In particular, since
\propref{maj-ampl} is proven by induction, and since the value of $\ell$ 
affects only the base case, the next lemma implies also the same kind of
monotonicity for general $k$.
\begin{lem}
\label{lem:cfr_bins}Let $k=2$, $a=1$, let $\ell$ be odd, and let $(\mathbf{c}\cdot P)_{1}\geq(\mathbf{c}\cdot P)_{2}$.
The rule of Stage 2 of the protocol is such that
\begin{align}
\Pr\left(\maj u\ell=1\right) & =  \Pr\left(\maj u{\ell+1}=1\right)\leq\Pr\left(\maj u{\ell+2}=1\right)\nonumber,\\
\Pr\left(\maj u\ell=2\right) & =  \Pr\left(\maj u{\ell+1}=2\right)\geq\Pr\left(\maj u{\ell+2}=2\right).
\label{eq:maj_monotone}
\end{align}
\end{lem}
\begin{proof}
To simplify notation, let $p_{1}=(\mathbf{c}\cdot P)_{1}$
and $p_{2}=(\mathbf{c}\cdot P)_{2}$. By definition, we
have 
\begin{align*}
\Pr\left(\maj u\ell=1\right) & =\Pr\left(X_{1}^{\left(\ell \right)} 
\geq \left\lceil \frac{\ell }{2}\right\rceil \right),\\
\Pr\left(\maj u{\ell+1}=1\right) & 
=\Pr\left(X_{1}^{\left(\ell +1\right)}>\frac{\ell +1}{2}\right) + 
\frac{1}{2}\Pr\left(X_{1}^{\left(\ell +1\right)}=\frac{\ell +1}{2}\right),\\
\Pr\left(\maj u{\ell+2}=1\right) & =\Pr\left(X_{1}^{\left(\ell +2\right)} 
\geq \left\lceil \frac{\ell +2}{2}\right\rceil \right),
\end{align*}

where $X_{1}^{(\ell )}$, $X_{1}^{(\ell +1)}$ and $X_{1}^{(\ell +2)}$
are binomial r.v. with probability $p_{1}$ and number of trials
$\ell ,$ $\ell +1$, and $\ell +2$, respectively. We can view $X_{1}^{(\ell )}$,
$X_{1}^{(\ell +1)}$, and $X_{1}^{(\ell +2)}$ as the
sum of $\ell $, $\ell +1$ and $\ell +2$ $Bernoulli(p_{1})$ r.v., respectively.
In particular,
let $Y$ and $Y^{'}$ be independent r.v. with distribution $Bernoulli(p_{1})$.
We can couple $X_{1}^{(\ell )}$, $X_{1}^{(\ell +1)}$
and $X_{1}^{(\ell +2)}$ as follows:
\[
X_{1}^{(\ell +1)}=X_{1}^{(\ell )}+Y
\]
and
\[
X_{1}^{(\ell +2)}=X_{1}^{(\ell +1)}+Y^{\prime}
\]
 Since $\ell $ is odd, observe that if $X_{1}^{(\ell )}>\left\lceil \frac{\ell }{2}\right\rceil $,
then $\maj u\ell=1$ regardless of the value of $Y$, and similarly if
$X_{1}^{(\ell )}<\left\lceil \frac{\ell }{2}\right\rceil $ then
$\maj u\ell=2$. Thus we have
\begin{align}
\Pr\left(\maj u{\ell+1}=1\right) &= \sum_{i=1}^{\ell }\Pr\left(\maj
u{\ell+1}=1\middle|\,X_{1}^{\left(l\right)}=i\right)\Pr\left(X_{1}^{\left(\ell
\right)}=i\right) \nonumber \\
&= \sum_{i>\left\lceil \frac{\ell }{2}\right\rceil }^{\ell }
\Pr\left(X_{1}^{\left(\ell \right)}=i\right)+
\Pr\left(\maj u{\ell+1}=1\middle|\,X_{1}^{\left(\ell\right)}
=\left\lceil \frac{\ell}{2}\right\rceil
\right)\Pr\left(X_{1}^{\left(\ell\right)}
=\left\lceil \frac{\ell}{2}\right\rceil \right) \nonumber \\
&+\Pr\left(\maj u{\ell+1}=1\middle|\,X_{1}^{\left(\ell\right)}
=\left\lfloor \frac{\ell}{2}\right\rfloor
\right)\Pr\left(X_{1}^{\left(\ell\right)}
=\left\lfloor \frac{\ell}{2}\right\rfloor \right).
\label{eq:majY'}
\end{align}
As for the last two terms in the previous equation, we have that 
\begin{gather}
\Pr\left(\maj u{\ell+1}=1\middle|\,X_{1}^{\left(\ell\right)}=\left\lceil \frac{\ell}{2}\right\rceil \right)=\Pr\left(Y=1\right)+\Pr\left(Y=0\right)\frac{1}{2},
\label{eq:maj_with_up}
\end{gather}
and
\begin{gather}
\Pr\left(\maj u{\ell+1}=1\middle|\,X_{1}^{\left(\ell\right)}=\left\lfloor \frac{l}{2}\right\rfloor \right)=\Pr\left(Y=1\right)\frac{1}{2}.
\label{eq:maj_with_down}
\end{gather}
Moreover, by a direct calculation one can verify that 
\begin{equation}
\Pr\left(X_{1}^{\left(\ell\right)}=\left\lceil \frac{\ell}{2}\right\rceil \right)=\frac{\Pr\left(Y=0\right)}{\Pr\left(Y=1\right)}\Pr\left(X_{1}^{\left(\ell\right)}=\left\lfloor \frac{\ell}{2}\right\rfloor \right).
\label{eq:bin_when_maj}
\end{equation}
From Eq.~(\ref{eq:maj_with_up}), (\ref{eq:maj_with_down}) and (\ref{eq:bin_when_maj})
it follows that 
\begin{align}
&\Pr\left(\maj u{\ell+1}=1\middle|\,X_{1}^{\left(\ell\right)}=\left\lceil
\frac{\ell}{2}\right\rceil
\right)\Pr\left(X_{1}^{\left(\ell\right)}=\left\lceil
\frac{\ell}{2}\right\rceil \right) \nonumber\\
&+\Pr\left(\maj u{\ell+1}=1\middle|\,X_{1}^{\left(\ell\right)}=\left\lfloor
\frac{\ell}{2}\right\rfloor
\right)\Pr\left(X_{1}^{\left(\ell\right)}=\left\lfloor
\frac{\ell}{2}\right\rfloor \right) \nonumber\\
&=\left(\Pr\left(Y=1\right)+\Pr\left(Y=0\right)\cdot\frac{1}{2}\right)
\Pr\left(X_{1}^{\left(\ell\right)}=\left\lceil
\frac{\ell}{2}\right\rceil
\right)+\left(\Pr\left(Y=1\right)\cdot\frac{1}{2}\right)
\Pr\left(X_{1}^{\left(\ell\right)}=\left\lfloor
\frac{\ell}{2}\right\rfloor \right) \nonumber\\
&=\Pr\left(X_{1}^{\left(\ell\right)}=\left\lceil \frac{\ell}{2}\right\rceil \right)\left(\Pr\left(Y=1\right)+\frac{\Pr\left(Y=0\right)}{2}+\frac{\Pr\left(Y=1\right)}{2}\frac{\Pr\left(Y=0\right)}{\Pr\left(Y=1\right)}\right)\nonumber\\
&=\Pr\left(X_{1}^{\left(\ell\right)}=\left\lceil \frac{\ell}{2}\right\rceil \right).
\label{eq:key_of_maj_YvsY'}
\end{align}
By plugging Eq.~(\ref{eq:key_of_maj_YvsY'}) in Eq.~(\ref{eq:majY'}) we get
\[
\Pr\left(\maj u\ell=1\right)=\Pr\left(\maj u{\ell+1}=1\right).
\]
The proof that 
\[
\Pr\left(\maj u\ell=2\right)=\Pr\left(\maj u{\ell+1}=1\right)
\]
is analogous, proving the first part of Eq.~(\ref{eq:maj_monotone}).

As for the second part, observe that if $X_{1}^{(\ell+1)}>\frac{\ell+1}{2}$,
then $\maj u{\ell+2}=1$ regardless of the value of $Y^{\prime}$, and
similarly if $X_{1}^{(\ell+1)}<\frac{\ell+1}{2}$ then $\maj u{\ell+2}=2$.
Observe also that 
\[
\Pr\left(\maj u{\ell+2}=1\middle|\,X_{1}^{\left(\ell+1\right)}=\frac{l+1}{2}\right)=\Pr\left(Y=1\right)=p_{1}.
\]
Because of the previous observations and the hypothesis that $p_{1}\geq\frac{1}{2}$,
we have that
\begin{align}
&\Pr\left(\maj u{\ell+2}=1\right) \nonumber \\
&=\sum_{i=0}^{\ell}\Pr\left(\maj u{\ell+2}=1\middle|\,X_{1}^{\left(\ell+1\right)}=i\right)\Pr\left(X_{1}^{\left(\ell+1\right)}=i\right)\nonumber \\
&=\sum_{i>\frac{\ell+1}{2}}^{l}\Pr\left(X_{1}^{\left(\ell+1\right)}=i\right)+\Pr\left(\maj u{\ell+2}=1\middle|\,X_{1}^{\left(\ell+1\right)}=\frac{\ell+1}{2}\right)\Pr\left(X_{1}^{\left(\ell+1\right)}=\frac{\ell+1}{2}\right)\nonumber \\
&=\sum_{i>\frac{\ell+1}{2}}^{\ell}\Pr\left(X_{1}^{\left(\ell+1\right)}=i\right)+p_{1}\cdot\Pr\left(X_{1}^{\left(\ell+1\right)}=\frac{\ell+1}{2}\right)\nonumber \\
&\geq\sum_{i>\frac{\ell+1}{2}}^{\ell}\Pr\left(X_{1}^{\left(\ell+1\right)}=i\right)+\frac{1}{2}\Pr\left(X_{1}^{\left(\ell+1\right)}=\frac{\ell+1}{2}\right)=\Pr\left(\maj u{\ell+1}=1\right).
\label{eq:key_proof_cfr_bin}
\end{align}
The proof of 
\[
\Pr\left(\maj u{\ell+2}=2\right)\leq\Pr\left(\maj u{\ell+2}=2\right)
\]
is the same up to the inequality in (\ref{eq:key_proof_cfr_bin}),
whose direction is reversed because $p_{2}\leq\frac{1}{2}$.
\end{proof}

\section{\label{apx:eps-smaller}Rumor spreading with $\epsilon=\Theta(n^{-\frac{1}{4}-\eta})$}

In \cite{FHK15} it is shown that at the end of
Stage 1 the bias toward the correct opinion is at least $\epsilon^{T+2}/2$ and,
at the beginning of Stage 2, they assume a bias toward the correct opinion of
$\Omega(\sqrt{\log n/n})$. In this section, we show that, when 
$\epsilon=\Theta(n^{-\frac{1}{4}-\eta})$ for some $\eta \in (0, 1/4)$, 
the protocol considered by \cite{FHK15} and us cannot solve the 
rumor-spreading and the plurality consensus problem in time 
$\Theta(\log n/\epsilon^{2})$.

First, observe that when $\epsilon=\Theta(\sqrt{\log n/n})$ the
length of the first phase of Stage 1 is $\Theta\left(\log
n/\epsilon^{2}\right)=\Omega(n\log n)$, which implies that, w.h.p.,
each node gets at least one message from the source during the first phase.
Thus, thanks to our analysis of Stage 2 we have that when
$\epsilon=\Theta(\sqrt{\log n/n})$ the protocol effectively solves
the rumor-spreading problem, w.h.p., in time $\Theta(\log n/\epsilon^{2})$.

In general, for $\epsilon < n^{-1/2-\eta}$ for some constant $\eta>0$, if we
adopt the second stage right from the beginning (which means that the source
node sends $\epsilon^{-2}$ messages), we get that, w.h.p., all nodes receive at
least $\log n /( \epsilon^2 n )$ messages. Thus,
by a direct application of \lemref{cb-diff}, after the first phase we get
an $\sqrt{ \log n / n}$-biased \config, w.h.p., and Stage 2 correctly 
solves the problem according to \thmref{general_majority_consensus}.

However, when $\epsilon=\Theta(n^{-\frac{1}{4}-\eta})$ for some
$\eta>0$, from \claimref{bootstrap} and \lemref{correct-stage1} we have that,
after phase 0 in \config $\mathbf{c}$, at most $\bigo \left(\log
n/\epsilon^{2}\right) =\bigo (n^{\frac{1}{2}+2\eta}\log n)$ nodes
are opinionated, and $\mathbf{c}$ is $\frac{\epsilon}{2}$-biased. Each node
that gets opinionated in phase 1 receives a message pushed from some node of
$\mathbf{c}$, and, because of the noise, the value of this message is distributed
according to $\mathbf{c}^{(\tau_{0})}\cdot P$. It follows that
$\mathbf{c}$ is an $\epsilon^{2}/2$-biased \config with
$\epsilon^{2}=n^{-\frac{1}{2}-2\eta}$ which is much smaller than the
$\Omega(\sqrt{\log n/n})$ bound required for the second stage. 

We believe that no minor modification of the protocol proposed here can
correctly solve the noisy rumor-spreading problem when $\epsilon =\Theta
(n^{-\frac{1}{4}-\eta})$ in time $\bigo \left( \log n/\epsilon^2
\right)$.
\end{document}